\documentclass[journal, draftcls, onecolumn]{IEEEtran}

\usepackage{amsmath, mathrsfs, graphicx,epsfig,amssymb, color, subfigure}

\newtheorem{lemma}{\bf Lemma}
\newtheorem{remark}{\bf Remark}

\newtheorem{thm}{\bf Theorem}

\newtheorem{rem}{\bf Remark}
\newtheorem{cor}{\bf Corollary}
\newtheorem{ex}{\bf Example}
\newtheorem{cl}{\bf Claim}

\begin{document}

\title{The Diversity-Multiplexing Tradeoff of the Dynamic Decode-and-Forward Protocol on a MIMO Half-Duplex Relay Channel}
%
%
%

\author{{\Large Sanjay ~Karmakar ~~~~~~Mahesh~ K. ~Varanasi}
\thanks{S. Karmakar and M. K. Varanasi are with the Department
of Electrical Computer and Energy Engineering, University of Colorado, Boulder,
CO, 30809 USA e-mail: (sanjay.karmakar@colorado.edu, varanasi@colorado.edu).}
\thanks{This work was supported in part by the US National Science Foundation Grants
ECCS-0725915 and CCF-0728955.}\\
\thanks{The material in this paper was presented in part at the IEEE International Symposium on Information Theory, 2009, Seoul and at the Asilomar Conference on Signals, Systems and Computers, Pacific Grove, CA, 2009}}

\markboth{Submitted, IEEE Trans. Inform. Th., August~2010}%
{Shell \MakeLowercase{\textit{et al.}}: Bare Demo of IEEEtran.cls for Journals}
%

\maketitle

\begin{abstract}
The diversity-multiplexing tradeoff of the dynamic decode-and-forward protocol is characterized for the half-duplex three-terminal $(m,k,n)$-relay channel where the source, relay and the destination terminals have $m$, $k$ and $n$ antennas, respectively. It is obtained as a solution to a simple, two-variable, convex optimization problem and this problem is solved in closed form for special classes of relay channels, namely, the $(1,k,1)$ relay channel, the $(n,1,n)$ relay channel and the $(2,k,2)$ relay channel. Moreover, the tradeoff curves for a certain class of relay channels, such as the $(m,k,n>k)$ channels, are identical to those for the decode-and-forward protocol for the full duplex channel while for other classes of channels they are marginally lower at high multiplexing gains. Our results also show that for some classes of relay channels and at low multiplexing gains the diversity orders of the dynamic decode-and-forward protocol protocol are greater than those of the static compress-and-forward protocol which in turn is known to be tradeoff optimal over all {\em static} half duplex protocols. In general, the dynamic decode-and-forward protocol has a performance that is comparable to that of the static compress-and-forward protocol which, unlike the dynamic decode-and-forward protocol, requires global channel state information at the relay node. Its performance is also close to that of the decode-and-forward protocol over the full-duplex relay channel thereby indicating that the half-duplex constraint can be compensated for by the dynamic operation of the relay wherein the relay switches from the receive to the transmit mode based on the source-relay channel quality.
\end{abstract}



%
\IEEEpeerreviewmaketitle

\newpage

\section{Introduction}
\label{intro}

Higher transmission rates and increased reliability or quality-of-service are two of the most important goals in the design of wireless communication systems. Techniques enabling the simultaneous realization of higher transmission rates and reliability include the employment of multiple antennas at the receiver and the transmitter and cooperation or relaying among users of the network. In this paper, our interest is on a three-terminal relay network with the source, relay and destination each equipped with multiple, and possibly distinct, number of antennas. One application that is being considered for relaying, for example, is the potential for expanded throughput and coverage for broadband wireless access (IEEE 802.16) with rapid and low cost deployment of relay stations of complexity and cost lower than that of legacy base stations but higher than that of mobile stations \cite{relayTG}.

Communication theoretic results on relaying in wireless channels can be found in \cite{SEB1,SEB2,LW,LanemanJN:coop,NabarRU:relay,PrasadN:CO-OP:ISIT04,KHP,NpV} for ergodic \cite{SEB1,SEB2} and outage settings \cite{LW,LanemanJN:coop,NabarRU:relay,PrasadN:CO-OP:ISIT04,KHP,NpV} with \cite{LW,LanemanJN:coop,PrasadN:CO-OP:ISIT04,KHP,NpV} characterizing the diversity-multiplexing tradeoff (DMT), a high SNR metric originally proposed for the multiple-input, multiple-output (MIMO) Rayleigh fading point-to-point links in \cite{tse1}, of several increasingly high performance half-duplex (HD) relaying protocols for single antenna terminals in \cite{LW,LanemanJN:coop,PrasadN:CO-OP:ISIT04,KHP,NpV} and for the case of a multiple-antenna destination in \cite{PrasadN:CO-OP:ISIT04}. Of these protocols, the one that concerns us in this work is the so-called dynamic decode-and-forward (DDF) protocol of \cite{KHP} but in the much more general context of a relay network with multiple and arbitrary number of antennas at each of the three nodes. The word ``dynamic" in dynamic decode-and-forward highlights the feature of this protocol wherein the relay listens for a source-to-relay-channel-dependent fraction of a frame before deciding to transmit to the destination. Protocols wherein the relay listens for an {\em a priori} fixed fraction of the frame length are called static protocols.

Relay networks with multiple antenna nodes were first considered in~\cite{YEE} where the authors analyzed the performance of a number of cooperative protocols and showed that the compress-and-forward (CF) protocol attains the fundamental DMT of both the full-duplex (FD) and the static half-duplex relay networks. Our choice of the DDF protocol however, is based on the fact that the CF protocol requires that the relay have perfect and global channel knowledge (i.e., the channel matrices between each of the three pairs of nodes) which may be difficult or even impossible to realize in practice. Moreover, while practical finite-length coding/decoding schemes based on the DDF protocol (cf. \cite{RC} and the references therein) exist, no corresponding code has been found -- to the best of our knowledge -- for the CF protocol. In contrast to the CF protocol, the DDF protocol requires the relay node to merely know its incoming channel. This suggests a possible performance-complexity tradeoff between the CF and DDF protocols which can be illuminated in the high SNR regime by providing the DMT achievable by both these protocols on the half duplex MIMO relay channel. While~\cite{YEE} proved the optimality of the static CF (SCF) protocol on the half-duplex relay channel, it did not provide an explicit DMT of the half-duplex channel under the constraint on protocols being static. More recently, explicit DMT characterizations of the SCF protocol and the DDF protocol for the so-called symmetric half-duplex MIMO relay channel, in which the number of antennas at the source and destination are equal, were reported in \cite{lvy_2009} and by the authors in a conference version of this paper in \cite{skv_2009}, respectively. Moreover, the work on the DDF protocol in \cite{skv_2009} was generalized to the relay channel with an arbitrary number of antennas at the three nodes in a second conference version of this paper in \cite{SkV2}, and independently and at almost the same time, the authors of \cite{lvy_2009} obtained a similar generalization of their work on the SCF protocol in \cite{OCM}. Furthermore, there is one key enabling analytical tool that is common to both works, namely the specification of the joint distribution of the eigenvalues of two specially correlated Wishart matrices, but the methods employed to solve this problem are different in \cite{OCM} and in this paper (first reported in \cite{SkV2}). Moreover, the dynamic nature of the DDF protocol considered in this work introduces another source of difficulty in the analysis that is not encountered in the analysis of the SCF protocol in \cite{OCM}.
Note that the generalization to the relay channel with an arbitrary number of nodes at the source, relay and destination is not only mathematically interesting but it is also a practically important problem. For example, this extra generality is critical in the application of relaying in broadband wireless access \cite{relayTG} where the three nodes are envisioned to have unequal number of antennas and computational capability. Other potential practical examples of cooperative networks that involve terminals with different numbers of antennas are also detailed in Section \ref{sec_explicit_dmt} along with a comparative DMT performance of the DDF protocol with non-cooperative communication as well as with full-duplex DF (FD-DF) and the SCF protocols.

As stated earlier, it was found in \cite{YEE} that the SCF protocol is DMT optimal on a relay channel under the constraint that the relay node operates statically. This doesn't of course preclude the DDF protocol from outperforming the SCF protocol since in the DDF protocol the relay operates in the dynamic mode. Indeed, comparing the DMT curves of the DDF protocol with that of the SCF protocol, it is found that for some channel configurations and at lower multiplexing gains, the DDF protocol does in fact achieve higher diversity orders than the SCF protocol. This proves that a half-duplex relay node operating via a static protocol prevents optimal performance over the HD channel. That the DDF protocol does not always perform uniformly better than the SCF protocol can be explained from the fact that the DF strategy is itself in general sub-optimal for static half-duplex and full duplex relaying \cite{YEE}. While allowing dynamic operation improves the DF strategy in low multiplexing gain regimes sometimes beyond even that of the SCF protocol, the superior performance of the SCF protocol over its DF counterpart persists in spite of allowing dynamic operation for high multiplexing gains. While performance improvement over the DDF protocol was sought within the framework of dynamic operation of the relay and the decode-and-forward strategy in \cite{NpV} for a relay channel with single antennas nodes, we do not pursue this improvement in this paper for the case of multiple antenna nodes.

Comparison with the DMT performance of the full duplex decode-and-forward (FD-DF) protocol also reveals an interesting fact. For a number of cases depending on the relative numbers of antennas at the three nodes, the optimal DMTs of the FD-DF and the DDF protocol can be nearly equal. In these cases therefore, the extra cost of full duplex relaying (due to enabling simultaneous transmission and reception) can be completely offset relative to half-duplex relaying by allowing dynamic operation.

It is also noteworthy that the application of the DDF protocol is not only limited to the relay channel. In \cite{KHP1}, it was shown that the DDF protocol is optimal on both a relay channel with automatic-retransmission-request (ARQ) protocol and a multiple-access-channel with a relay (MAR) and ARQ, with single antenna nodes. This encourages one to further analyze the performance of the protocol on these channels with multiple antenna nodes. The performance analysis of the DDF protocol on a MIMO half-duplex three node relay channel can be seen to provide the first step in that direction. 

The rest of the paper is organized as follows. In Section~\ref{section_system_model}, we describe the system model and the DDF protocol. In Section \ref{sec_eigenvalue_distribution}, we provide the eigenvalue distribution result using which, in Section \ref{sec_outage_and_error_analysis}, we derive the outage probability of the DDF protocol and specify the optimization problem whose solution is its DMT. In Section \ref{sec_closed_form_dmt_calculations}, closed-form solutions for three simple channel configurations are provided, following which explicit DMT curves are provided using these methods for a few more channel configurations in Section~\ref{sec_explicit_dmt}. Section \ref{conclusion} concludes the paper.

\begin{proof}[Notations]$(x)^+$, $x\land y$, $| \mathcal{X} | $ $|X|$ and $(X)^{\dagger}$ represent $\max \{0,x\}$, the minimum of $x$ and $y$, the size of the set $\mathcal{X}$, the determinant, and the conjugate transpose of the matrix, $X$, respectively. Let $\mathbb{R}$ and $\mathbb{C}$ denote the real and complex number fields, respectively, and $\mathbb{C}^{n\times m}$ the set of all $n\times m$ matrices with complex entries. The interval containing all real numbers between $x$ and $y$ will be denoted by $[x,y]$, i.e., $[x,y]=\{z\in \mathbb{R}:x\leq z\leq y\}$. Similarly we denote the set $\{z\in \mathbb{R}:x < z\leq y\}$ by $(x,y]$. The empty set is denoted by $\Phi$. Let $[a_{i,j}]_{i,j=1}^{M_1,N_1}$ represent a matrix in $\mathbb{C}^{M_1\times N_1}$, where $a_{i,j}$ represents the element in the $i^{th}$ row and $j^{th}$ column. If $x_1, x_2, \cdots , x_u$ represents a set of real numbers then $\bar{x}$ represents the vector whose components are $x_i$s, i.e., $\bar{x}=[x_1, x_2, \cdots , x_u]$. The Vandermonde matrix $[x_i^{(j-1)}]_{i,j=1}^{u,u}$ formed from the vector $\bar{x}=[x_1, x_2, \cdots , x_u]$ will be denoted by $\mathbf{V}_1(\bar{x})$. The probability distribution of a complex Gaussian random variable with zero mean and unit variance is denoted by $\mathcal{CN}(0,1)$. The symbol $\textrm{diag}(.)$ represents a square diagonal matrix of corresponding size with the elements in its argument on the diagonal and $I_n$ denotes an $n\times n$ identity matrix.  The probability of an event $\mathcal{E}$ is denoted as $Pr(\mathcal{E})$. All the logarithms in this text are to the base $2$. Finally, any two functions $f(\rho)$ and $g(\rho)$ of $\rho$, where $\rho$ is the signal-to-noise ratio (SNR) defined later, are said to be exponentially equal and denoted as $f(\rho)\dot{=}g(\rho)$ if,
\begin{equation*}
\lim_{\rho \to \infty} ~\frac{\log(f(\rho))}{\log(\rho)} ~~= ~~\lim_{\rho \to \infty} ~\frac{\log(g(\rho))}{\log(\rho)},
\end{equation*}
$\dot{\leq}$ and $\dot{\geq }$ signs are defined similarly. We also define the following function
\begin{equation}
\label{varphi-fn}
\varphi(x,y)=\left\{\begin{array}{cc}
0,~\textrm{if}~x<y;\\
+\infty,~\textrm{if}~x\geq y.
\end{array}\right.
\end{equation}
\end{proof}

\begin{figure}[!thb]
\setlength{\unitlength}{1mm}
\begin{picture}(80,40)
\thicklines
\put(50,10){\circle{10}}
\put(110,10){\circle{10}}
\put(80,32){\circle{10}}
\thicklines
\put(60,10){\vector(1,0){40}}
\put(55,15){\vector(4,3){18}}
\put(87,28){\vector(4,-3){18}}
\put(49,9){$\mathbf{S}$}
\put(109,9){$\mathbf{D}$}
\put(79,31){$\mathbf{R}$}
\put(48,18){$(m)$}
\put(88,32){$(k)$}
\put(108,18){$(n)$}
\put(57,24){$\mathbf{H_{SR}}$}
\put(94,24){$\mathbf{H_{RD}}$}
\put(78,12){$\mathbf{H_{SD}}$}
\end{picture}
\caption{System model of the $(m,k,n)$ MIMO relay channel.}
\label{system_model}
\end{figure}
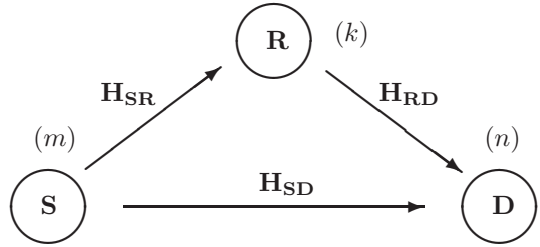

\section{System Model}
\label{section_system_model}
Consider a quasi-static Rayleigh faded MIMO relay channel with a single relay node as shown in Figure \ref{system_model}, where the source, the destination and the relay node have $m$, $n$ and $k$ antennas, respectively. Let $H_{SR}\in\mathbb{C}^{k\times m}$, $H_{SD}\in\mathbb{C}^{n\times m}$ and $H_{RD}\in\mathbb{C}^{n\times k}$ represent the channel matrices between source and relay, source and destination and relay and destination, respectively. For economy of notation these channel matrices will be denoted by $H$ collectively, i.e., $H=\{H_{SR},~H_{SD},~H_{RD}\}$. The quasi-static fading assumption implies that these channel coefficient matrices remain fixed for the entire duration of a codeword and change independently from one codeword to the next. All these matrices are assumed to be mutually independent and the elements of these matrices are independently and identically distributed (i.i.d.) as $\mathcal{CN}(0,1)$, thus modeling Rayleigh fading. Let the channel state information be perfectly known at the receivers but unknown at the transmitters. Suppose that independent random Gaussian codes are used by both the source and the relay node.

\subsection{The DDF protocol}

The DDF protocol was proposed and analyzed in~\cite{KHP} for an HD relay channel with {\em single-antenna} nodes.  In this protocol, the relay node has two phases of operation. In the first phase, the relay node listens to the source transmission and decodes it as soon it receives enough mutual information to do so. In particular, if $\hat{l}$ is the minimum integer such that $\hat{l}\log(\det(I_n+\rho {H_{SR}}^{\dagger}H_{SR}))\geq l R$, where $R$ is the rate of transmission in bits per channel use, $\rho$ is the signal-to-noise ratio (SNR) of the source to relay link and $l$ is the block length of the source codeword, then the first phase ends at the $\hat{l}$-th channel use. During the first phase the relay node does not transmit.  The second phase starts from the $(\hat{l}+1)$-st channel use and lasts for the rest of the source transmission (i.e., it consists of $l-\hat{l}$ channel uses). In what follows, $\hat{l}$ is called the relay {\em decision time} as in \cite{RC}. During the second phase, the relay re-encodes the source message using an independent codebook and transmits it during the rest of the codeword. Note that the relay node can help by transmitting an independent copy of the message to the destination only if $\hat{l}<l$. Otherwise, it does not participate in the cooperation and the channel behaves like a point-to-point (PtP) channel. Clearly, in this scheme, the source codeword $X_S$ consists of two parts ($X_S=[X_{S1}, X_{S2}]$), the first part ($X_{S1}\in \mathbb{C}^{m\times (\hat{l}\land l)}$) is sent by the source while the relay is listening and the second part ($X_{S2}\in \mathbb{C}^{m\times (l-(\hat{l}\land l)}$) is transmitted while the relay is transmitting its own codeword $X_R\in \mathbb{C}^{k\times (l-(\hat{l}\land l))}$. Thus the received signal at the relay and the destination in phase one can be written as
\begin{equation*}
\begin{array}[r]{l}
Y_{1D}=\sqrt{\rho} H_{SD}X_{S1}+N_1, \qquad Y_{1D}\in\mathbb{C}^{n\times (\hat{l}\land l)},\\
Y_{R}=\sqrt{\rho} H_{SR}X_{S1}+N_R, \qquad Y_{R}\in\mathbb{C}^{k\times (\hat{l}\land l)},
\end{array}
\end{equation*}
and the received signal at the destination in phase two is given by
\begin{equation*}
Y_{2D}=\sqrt{\rho} H_{SD}X_{S2}+\sqrt{\rho} H_{RD}X_R+N_2, \qquad Y_{2D}\in\mathbb{C}^{n\times (l-(\hat{l}\land l))},
\end{equation*}
where $N_1\in\mathbb{C}^{n\times (\hat{l}\land l)}$, $N_2\in\mathbb{C}^{n\times (l-(\hat{l}\land l))}$ and $N_{R}\in\mathbb{C}^{k\times (\hat{l}\land l)}$ represent the additive noises at the destination during the first and second phases and at the relay, respectively. All the entries of $N_1$, $N_2$ and $N_R$ are assumed to be i.i.d. $\mathcal{CN}(0,1)$. Besides assuming channel state information at the receivers, we assume for simplicity, as does \cite{KHP}, that the destination has perfect (genie-aided) knowledge of the relay decision time $\hat{l}$ which of course is a function of the source-to-relay channel $H_{SR}$\footnote{The assumption of genie-aided relay decision time (and infinite codeword length) is relaxed and addressed rigorously for the single-antenna relay channel in the recent work of \cite{RC} where it is shown that there is no loss of diversity-multiplexing tradeoff optimality if the relay does not convey the side information about relay decision time but the decoder at the destination {\em jointly} decodes the decision time and the message.}. Further, to ensure that $\rho$ represents the SNR of each link, we impose the following constraints on the covariance matrices of the inputs:
\begin{equation}
E\left(X_{S,T} X_{S,T}^\dagger\right)=I_{m\times m} ~\textrm{and}~E\left(X_{R,T} X_{R,T}^\dagger\right)=I_{k\times k}, ~\forall T,
\end{equation}
where $X_{S,T}$ and $X_{R,T}$ represent the $T$-th column of the source and the relay codeword, respectively. Let us also define the ratio $\frac{\min\{\hat{l},l\}}{l}$ by $f$, i.e., $f$ represents the fraction of time for which the relay node listens before starting its own transmission, if it can decode the source transmission. In Section~\ref{sec_outage_and_error_analysis}, we shall see that this parameter $f$ plays an important role in the formulation of an appropriate outage event. Note that $f$ is defined in terms of $\hat{l}$ which is a function of the rate of transmission $R$, and source-to-relay channel $H_{SR}$, making it a random variable. In what follows, we derive the dependence of $f$ on the channel matrices more rigorously. From the definition of $\hat{l}$ specified earlier, we get
\begin{IEEEeqnarray}{rl}
\hat{l} =&\left\lceil \frac{lR}{\log(\det(I_n+\rho {H_{SR}}^{\dagger}H_{SR}))}\right\rceil
 \nonumber\\
\label{ratio}
=&\left\lceil \frac{lr \log(\rho)}{\log(\rho)\sum_{i=1}^{t}(1-\gamma_i)^{+}}\right\rceil,~\left[\because~ R=r\log(\rho)\right] \nonumber\\
\end{IEEEeqnarray}
where $\rho^{-\gamma_i}=\nu_i,~1\leq i\leq \min\{k, m\}=t$, $\nu_1\geq \nu_2 \geq \cdots \geq \nu_t\geq 0$\footnote{Note that $\nu_t>0$ with probability 1 (w.p.1).} are the ordered eigenvalues of the central Wishart matrix $H_{SR}^{\dagger}H_{SR}$ and $r$ is the multiplexing gain. Putting this into the definition of $f$, in the limit when $ l \to \infty $, we get
\begin{IEEEeqnarray}{l}
\label{eq_f_dependence_on_channel}
f=\frac{\min\{\hat{l},l\}}{l}=\min \left\{1,\frac{r}{\sum_{i=1}^{t}(1-\gamma_i)^{+}}\right\}.
\end{IEEEeqnarray}

Besides $f$, computation of the outage probability will also involve the joint distribution of 2 Wishart matrices mutually correlated in a special way. In the next section, we shall describe the structure of this correlation and compute the corresponding joint distribution.

\section{Joint Eigenvalue distribution of two correlated matrices}
\label{sec_eigenvalue_distribution}
In the DMT analysis we need only the asymptotic (in SNR) distribution of the eigenvalues of the matrices appearing in the outage formulation. In this section, we shall derive the joint distribution of the eigenvalues of two such mutually correlated Wishart matrices. Mathematically, the asymptotic behavior of the eigenvalues of a random matrix is captured, following  \cite{tse1}, as shown below. Letting the ordered eigenvalues of a matrix of interest be denoted by $\pi_1\geq \pi_2\geq \cdots \geq \pi_u$, the asymptotic nature of the eigenvalues is characterized by $\delta_i$'s, where
\begin{equation}
\label{eq_asymptotic_def}
\pi_i=\rho^{-\delta_i},~1\leq i\leq u.
\end{equation}
Eventually, in this section we shall derive the joint distributions of these $\delta_i$'s; the following Theorem is the first step in that direction.

\vspace{2mm}

\begin{thm}
\label{thm_eigenvalue_distribution_1}
Let $H_1 \in \mathbb{C}^{N_2\times N_1}$ and $H_2 \in \mathbb{C}^{N_2\times N_3}$ be two mutually independent random matrices with i.i.d. $\mathcal{CN}(0,1)$ entries. Suppose that $\xi_1 \geq \xi_2\geq \cdots \xi_q > 0$ and $\lambda_1 \geq \lambda_2 \geq \cdots \lambda_p > 0$ are the ordered non-zero eigenvalues (w.p.$1$) of $V_1\triangleq H_1^{\dagger}(I_{N_2}+\rho H_2H_2^{\dagger})^{-1}H_1$ and $V_2\triangleq H_2H_2^{\dagger}$, respectively, with $p=\min\{N_2, N_3\}$ and $q=\min\{N_1, N_2\}$, and where all the eigenvalues are assumed to vary exponentially with SNR in the sense of equation \eqref{eq_asymptotic_def}. Then, the conditional asymptotic probability density function (pdf) of the eigenvalues $\bar{\xi}$ given $\bar{\lambda}$ is given as
\begin{eqnarray}
\mathbf{f_1}(\bar{\xi}|\bar{\lambda}) \, \dot{=} \, \prod_{j=1}^{q}(\xi_j^{(N_1+N_2-2j)}e^{-\xi_j}) \prod_{\substack{(u=1,v=1)\\ ((u+v)= (N_2+1))}}^{(p,q)} \left({e}^{-\rho \xi_v \lambda_u}\right) 
\prod_{i=1}^{p}(1+\rho \lambda_i)^{N_1}  \prod_{j=1}^{q} \prod_{i=1}^{(N_2-j)\land N_3} \left(\frac{1-{e}^{-\rho \xi_j \lambda_i}}{\rho \xi_j \lambda_i}\right). \nonumber
\end{eqnarray}
\end{thm}

\begin{proof}[Proof~(Outline)]
Let the singular-value-decomposition (SVD) of $V_2$ be $V_2=U^{\dagger}\Lambda U$, where $U\in \mathbb{C}^{N_2\times N_2}$ is a unitary matrix and $\Lambda \triangleq\textrm{diag}([\lambda_1, \lambda_2 \cdots \lambda_{N_2}])$, where $\lambda_1 \geq \lambda_2 \geq \cdots \lambda_{N_2}$ are the eigenvalues of $V_2$ (Note that $(N_2-p)$ of these are $0$ w.p.1). Denoting $\Sigma=(I_{N_2}+\rho \Lambda)^{-1}$ and $\hat{H}_1=U H_1$, $V_1$ can be written as $V_1=\tilde{H}_1^\dagger \tilde{H}_1$ where $\tilde{H_1}={\Sigma}^{1/2}\hat{H}_1$. $\tilde{H_1}$ can be thought of as a channel matrix of $N_1$ transmit antennas and $N_2$ receive antenna MIMO channel, where the channel is correlated at the receiver only, with the covariance matrix being $\Sigma$. The eigenvalue distribution of $V_1$ for such a channel was derived in \cite{GaS} for $N_2>N_1$ and in \cite{CZZ} for $N_2\leq N_1$, respectively. However, the expressions for the corresponding distributions given in \cite{GaS} and \cite{CZZ} can not be used directly for the DMT calculation as they involve ratio of determinants whose components are hypergeometric functions. Fortunately, for high SNR values these expressions can be simplified. A detailed proof is given in Appendix~\ref{pf_thm_eigenvalue_distribution_1}.
\end{proof}

\begin{cor}
\label{joint_eigen}
The joint pdf of $\bar{\xi}$ and $\bar{\lambda}$ is given as
\begin{eqnarray}
\label{joint_distribution}
\mathbf{f}(\bar{\xi},\bar{\lambda})\dot{=}\prod_{i=1}^{p}\left((1+\rho \lambda_i)^{N_1}e^{-\lambda_i}{\lambda_i}^{(N_2+N_3-2i)}\right) \prod_{\substack{(u=1,v=1)\\ ((u+v)= (N_2+1))}}^{(p,q)} \left({e}^{-\rho \xi_v \lambda_u}\right) \nonumber{} \\
\label{eq_joint_distribution_11}
 \prod_{j=1}^{q}(\xi_j^{(N_1+N_2-2j)}e^{-\xi_j}) \prod_{j=1}^{q} \prod_{i=1}^{(N_2-j)\land N_3} \left(\frac{1-{e}^{-\rho \xi_j \lambda_i}}{\rho \xi_j \lambda_i}\right).
\end{eqnarray}
\end{cor}

\begin{proof}[{\bf Proof of Corollary~\ref{joint_eigen}}]
The joint distribution of the ordered eigenvalues of $V_2=H_2H_2^{\dagger}$ is given in \cite{CZZ}, which for asymptotically high $\rho$ values becomes
\begin{equation*}
\mathbf{f_2}(\bar{\lambda})\dot{=}\prod_{i=1}^{p}e^{-\lambda_i} {\lambda_i}^{(N_2+N_3-2i)}.
\end{equation*}
 Using this marginal distribution of $\bar{\lambda}$ along with the conditional distribution of Theorem~\ref{thm_eigenvalue_distribution_1} we get \eqref{eq_joint_distribution_11}.
\end{proof}

Now, using the transformations $\lambda_i=\rho^{-\alpha_i}$ for $1\leq i \leq p$ and $\xi_j=\rho^{-\beta_j}$ for $1\leq j\leq q$ in equation~(\ref{joint_distribution}) we get the following
\begin{thm}
\label{thm_eigenvalue_distribution_2}
If the non-zero ordered eigenvalues of $V_1=H_1^{\dagger}(I_{N_2}+\rho H_2H_2^{\dagger})^{-1}H_1$ and $V_2=H_2H_2^{\dagger}$ are denoted by $\xi_j=\rho^{-\beta_j}, ~1\leq j\leq q$ and $\lambda_i=\rho^{-\alpha_i}, ~1\leq i\leq p$, respectively, where $H_1$ and $H_2$ are as in Theorem~\ref{thm_eigenvalue_distribution_1}, then the joint distribution of $\bar{\alpha}$ and $\bar{\beta}$ is given by
\begin{equation}
\label{distribution2}
\mathbf{g}(\bar{\alpha},\bar{\beta})\dot{=}\left\{
\begin{array}{l}
\rho^{-E\left(\bar{\alpha},\bar{\beta}\right)},~ \textrm{if}~(\bar{\alpha},\bar{\beta})\in \mathcal{A};\\
0,~ \textrm{if}~(\bar{\alpha},\bar{\beta})\notin \mathcal{A},
\end{array}\right.
\end{equation}
where $\mathcal{A}=\Big\{(\bar{\alpha},\bar{\beta}):(\alpha_i+\beta_j)\geq 1,~ \forall (i+j)\geq(N_2+1);~
 0\leq \alpha_1 \leq \cdots \leq \alpha_p;~
 0\leq \beta_1 \leq \cdots \leq \beta_q \Big\}$ and
\begin{IEEEeqnarray}{rl}
E\left(\bar{\alpha},\bar{\beta}\right)=\sum_{i=1}^{p}\Big((N_3+N_2-2i+1)\alpha_i- N_1(1-\alpha_i)^+\Big)+ & \sum_{j=1}^{q}(N_1+N_2-2j+1)\beta_j\nonumber \\
&+\sum_{j=1}^{q}\sum_{i=1}^{(N_2-j)\land N_3}(1-\alpha_i-\beta_j)^+.
\end{IEEEeqnarray}
\end{thm}
\begin{proof}
The proof is given in Appendix~\ref{pf_thm_eigenvalue_distribution_2}.
\end{proof}

This joint pdf of $\left(\bar{\alpha},\bar{\beta}\right)$ will be used in the next section to compute the probability of an appropriately defined outage event. The scope of the asymptotic joint pdf derived in this section however is not restricted to the results derived in this paper. Although the correlation between the two Wishart matrices has a specific structure, it may arise in different communication problems. For example, a similar correlation structure is encountered in the outage analysis of the 2-user Z interference channel and the result of this paper was used to derive the DMT of that channel in~\cite{Sanjay_Varanasi_ZIC_DMT}.

\section{DMT of the DDF protocol}
\label{sec_outage_and_error_analysis}
The optimal diversity order of a coding scheme, at a given multiplexing gain, is defined as the negative SNR exponent of the average codeword error probability averaged over the channel realizations. Thus to derive the DMT of a coding scheme it is important to first compute the average codeword error probability. In this section, we shall derive the best achievable diversity order of the DDF protocol in the following two steps: first, we shall show that the probability of error is exponentially equal to the probability of an appropriately defined outage event, $\mathcal{O}$; and then we shall compute the negative SNR exponent of this outage probability, $\Pr(\mathcal{O})$. It is the second step, where we shall have to use the distribution result derived in the previous section.

Let the average probability of codeword error of the DDF protocol, achievable over the MIMO relay channel at a given SNR, $\rho$ and minimized over all possible coding schemes, be denoted by $P^*_E(\rho)$, i.e.,
\begin{equation}
P^*_E(\rho)=\min_{\left\{\mathcal{C}(\rho)\in \mathscr{C}(\rho)\right\}} P^{\mathcal{C}(\rho)}_E,
\end{equation}
where $P^{\mathcal{C}(\rho)}_E$ represents the probability of codeword error achievable by a particular coding scheme $\mathcal{C}(\rho)$ and $\mathscr{C}(\rho)$ represents the family of possible codes at SNR $\rho$. Then the optimal diversity order, denoted by $d^*(r)$, at a multiplexing gain of $r$ is defined as
\begin{equation}
\label{def_optimal_diversity_order}
d^*(r)=\lim_{\rho \to \infty} ~\frac{-\log\left(P^*_E(\rho)\right)}{\log(\rho)}.
\end{equation}

\subsection{Probability of codeword error}

The computation of the average codeword error probability $P^*_E$ of the DDF protocol is divided in two parts depending on whether the relay node participates in the end-to-end communication or not. First, we consider the case where the relay node helps by cooperating (when $f<1$) and then we consider the case where the relay node does not participate in the communication (when $f=1$). We start with the first case.
\begin{itemize}
\item {\it Achievability $\left(P_E^*\dot{\leq} \Pr(\mathcal{O})\right)$\footnote{\textrm{The outage event $\mathcal{O}$ will be defined shortly.}}}:
\end{itemize}
Let us assume that the source and relay use independent Gaussian codebooks and denote the conditional codeword error probability by $P_{E|H}$, where
\begin{equation}
\label{eq_average_error_probability}
P_{E|H}\triangleq \sum_{\left\{\mathcal{C}_{G}(\rho)\right\}}\sum_{\left\{X_{S},\hat{X}_S\in \mathcal{C}_{G_s}(\rho)\right\}}\Pr\left(\mathcal{C}_{G}(\rho)\right)\Pr(\hat{X}_S,X_S)\Pr(X_S\neq \hat{X}_{S}|X_S,\hat{X}_S, H).
\end{equation}
That is the error probability is computed conditioned on a channel realization, $H$ and averaged over all ensembles of Gaussian codebooks having codeword length $l$ and cardinality $2^{lr\log(\rho)}$. This error probability can be upper bounded using Bayes' rule as,
\begin{IEEEeqnarray}{rl}
\label{bayes1}
P_{E|H}&=P_{E,E_r^c|H}+P_{E,E_r|H} \nonumber\\
&\leq P_{E|E_r^c,H}+P_{E_r|H},
\end{IEEEeqnarray}
where $E_r$ and $E_r^c$ represent the events of relay error and its complement. We know from equation \eqref{ratio} that the assumption $f<1$ is equivalent to saying that the source-to-relay link is not in outage. On the other hand, on a delay limited point-to-point (PtP) fading channel the best achievable probability of error is essentially equal to the so called information outage probability. Therefore, since we assume sufficiently large block length for the codewords used by the source, it can be easily proved that $P_{E_r|H}\leq \epsilon,$ for any $\epsilon >0$ and $f<1$.
\begin{rem}
The fact that $P_{E_r|H}\leq \epsilon$, for any $\epsilon >0, ~f<1$ and $l\to \infty$ was proved in \cite{KHP1} (Lemma~$1$) for a relay channel with single antenna nodes. The proof for the MIMO case is identical. In what follows, we provide an outline of the proof. Suppose, a codeword is divided into $N$ segments of length $L$ each, i.e. $l=LN$, where both $L$ and $N$ both grow to infinity and the ML decoder at the relay waits for $\hat{N}$ such segments to decode the message from the source, where $\hat{N}$ is given as
\begin{equation}
\hat{N}=\left\lfloor\frac{NR}{I(X_{S,T};Y_{R,T})}\right\rfloor+1,
\end{equation}
where $I(X_{S,T};Y_{R,T})$ represents the mutual information between the source and relay node at time $T$. Note that this $I(X_{S,T};Y_{R,T})$ is same for all $T$ since the input and noise are identically distributed across time and the channel is fixed for the entire codeword. From equation \eqref{ratio} and the fact that $f<1$ we get that $NR<\hat{N} I(X_{S,T};Y_{R,T})$. Now, $P_{E_r|H}$ as defined before represents the conditional probability of error on the PtP channel from the source to the relay node. $P_{E_r|H}$ can be upper bounded by replacing the ML decoder by a typical set decoder and then taking the average over the ensemble of codebooks. Then, following a method similar to that in in \cite{CT} (Theorem 10.1.1) it can be shown that
\begin{IEEEeqnarray}{l}
P_{E_r|H}\leq 2\epsilon+2^{3l\epsilon} 2^{-L(\hat{N}I(X_{S,t};Y_{R,t})-NR)}\leq 3\epsilon,
\end{IEEEeqnarray}
for sufficiently large $L$ and any $\epsilon>0$.
\end{rem}
Using this fact in equation \eqref{bayes1} and averaging both sides with respect to channel coefficients, in the high SNR limit we get
\begin{equation}
\label{asymptotic1}
P_{E}~\dot{\leq} ~\mathbb{E}_H\left(P_{E|E_r^c,H}\right) ~\dot{=} ~P_{E|E_r^c}.
\end{equation}
By the preceding argument, after $\hat{l}$ channel uses the relay node can decode the source message, where $\hat{l}=fl<l$. Suppose the relay node encodes the message into a codeword from its own codebook and starts transmitting it from the $(\hat{l}+1)$-th symbol. Thus, for the first $\hat{l}$ channel uses the relay channel essentially behaves like an $m\times n$ point-to-point channel and for the rest $(l-\hat{l})$ channel uses it behaves like an $(m+k)\times n$ point-to-point channel. Since the source and the relay use independent random Gaussian codes, averaging over the ensemble of random Gaussian codes, it can be easily proved~\cite{TB} that the pairwise error probability for a given channel realization, $P_{PE|E_r^c,H}$ is upper bounded as follows:
\begin{IEEEeqnarray}{l}
\label{eq_pep_1}
P_{PE|E_r^c,H}\leq \det\left(I_n+\frac{\rho}{2n}H_{SD}^{\dagger}H_{SD}\right)^{-\hat{l}}
\det\left(I_n+\frac{\rho}{2n}(H_{SD}^{\dagger}H_{SD}+H_{RD}^{\dagger}H_{RD})\right)^{-(l-\hat{l})}
\end{IEEEeqnarray}
\begin{rem}
The subtle difference between $P_{PE|E_r^c,H}$ and $P_{E|E_r^c,H}$, as defined in \eqref{eq_average_error_probability} should be noted. In the former, the averaging within a particular codebook is not done. However, the two can be related through the well known union bound as follows
\begin{equation*}
   P_{E|E_r^c,H}\leq |\mathcal{C}| P_{PE|E_r^c,H},
\end{equation*}
where $|\mathcal{C}|$ represents the cardinality of the codebooks.
\end{rem}
Recall that the cardinality of the codebooks were assumed to be $2^{rl\log(\rho)}={\rho}^{lr}$. Thus, using the union bound of probability of error in equation \eqref{eq_pep_1} we get
\begin{IEEEeqnarray}{rl}
\label{pairwise_error3}
P_{E|E_r^c,H}\leq {\rho}^{lr} P_{PE|E_r^c,H} =&\left[\det\left(I_n+\frac{\rho}{2n}H_{SD}^{\dagger}H_{SD}\right)^{-f}\right. \nonumber\\
& \left.  ~~~~~\times {\rho}^{r} \det \left(I_n+\frac{\rho}{2n}(H_{SD}^{\dagger}H_{SD}+H_{RD}^{\dagger}H_{RD})\right)^{-(1-f)}\right]^l
\end{IEEEeqnarray}
Now, if we define $\mathcal{O}$ in the following way,
\begin{IEEEeqnarray}{rl}
\label{outage_definition}
\mathcal{O}\triangleq\Bigg\{(\bar{\gamma},H_{SD},H_{RD}):& I(H) \triangleq f\log\left(\det(I_n+\frac{\rho}{2n}H_{SD}^{\dagger}H_{SD})\right)  \nonumber\\
&  +(1-f)\log\left(\det (I_n+\frac{\rho}{2n}(H_{SD}^{\dagger}H_{SD}+H_{RD}^{\dagger}H_{RD}))\right)\leq r\log({\rho})\Bigg\},
\end{IEEEeqnarray}
then it is evident from equation \eqref{pairwise_error3} that
\begin{equation}
\label{eq_prob_on_ontage_complement}
P_{E|\mathcal{O}^c,E_r^c}\to 0, ~\textrm{as} ~l\to \infty.
\end{equation}
Finally, from equation \eqref{asymptotic1} we have
\begin{IEEEeqnarray}{rl}
P_{E} ~\dot{\leq}&~P_{E|E_r^c},\nonumber\\
\stackrel{(a)}{=}&~ P_{E|\mathcal{O}^c,E_r^c}\Pr(\mathcal{O}^c)+P_{E|\mathcal{O},E_r^c} \Pr(\mathcal{O}),\nonumber\\
\leq &~P_{E|\mathcal{O}^c,E_r^c}\Pr(\mathcal{O}^c)+\Pr(\mathcal{O}), \nonumber\\
\label{p_upper_0}
\dot{\leq}& ~\Pr(\mathcal{O}),
\end{IEEEeqnarray}
where step $(a)$ follows from Bayes' rule and in the last step we used equation \eqref{eq_prob_on_ontage_complement}. Since $P_E$ represents the average probability of error averaged over ensemble of codes, there exist a code for which \eqref{p_upper_0} is true. Denoting the average probability of error for such a code by $P_e$ where the averaging is now over only the fading states, we have
\begin{equation}
\label{p_upper}
P_E^*\stackrel{(a)}{\leq}P_e\dot{\leq} \Pr(\mathcal{O})
\end{equation}
where step $(a)$ in the above equation follows from the fact that $P^*_E$ represents the minimum probability of error among all possible coding schemes and in the preceding analysis we have only considered Gaussian codes. The above equation establishes an upper bound on $P^*_E$. Next we derive a lower bound on $P^*_E$.

\vspace{2mm}

\begin{itemize}
\item {\it Converse ($P_E^*\dot{\geq} \Pr(\mathcal{O})$):}
\end{itemize}

Consider a genie aided relay channel where the genie gives the source message to the relay node after $\hat{l}$ channel uses. In the presence of such a genie the relay channel becomes a composite point-to-point channel, where for the latter $(1-f)$ fraction of the codeword, the relay and the source node together acts as the composite source. Clearly, $I(H)$ represents the mutual information between the source and the destination node and consequently $\mathcal{O}$ represents the outage event of the genie aided composite point-to-point MIMO channel. Thus using Fano's inequality as in~\cite{tse1} and the fact that the real system has a larger error probability than the genie-aided one, we get
\begin{equation}
\label{p_lower}
\Pr(\mathcal{O}) ~\dot{\leq } ~P^*_e(\textrm{genie})=\min_{\textrm{all coding schems}}P_e(\textrm{genie}) ~\dot{\leq } \min_{\textrm{all coding schems}}~P_e~= ~P^*_E,
\end{equation}
where $P_e$ and $P_e(\textrm{genie})$ represent the probability of error of the actual and genie aided system for any particular coding scheme. Finally, combining (\ref{p_lower}) and (\ref{p_upper}) we get
\begin{equation}
\label{eq_th1_case1}
    P_E^*\dot{=}\Pr\left(\mathcal{O}\right),~\textrm{for}~f<1.
\end{equation}

Next we consider the case when $f=1$. From the definition of $f$ in equation \eqref{ratio} we know that when $f=1$, the relay node does not take any part in the communication from the source to the destination. In this case, the relay channel becomes a point-to-point MIMO channel. It was shown in~\cite{tse1} that $P_E^*$ of such a channel is exponentially equal to the corresponding outage probability. Putting $f=1$ in our definition of $\mathcal{O}$ we get
\begin{equation}
\label{eq_outage_event_f1}
\mathcal{O}_{f=1}=\left\{H_{SD}: \log\left(\det(I_n+\frac{\rho}{2n}H_{SD}^{\dagger}H_{SD})\right) < r\log(\rho)\right\}.
\end{equation}
This is same as the outage event defined in~\cite{tse1} for a point-to-point channel having channel matrix $H_{SD}$ and thus using the result of \cite{tse1} we get
\begin{equation}
    P_E^*\dot{=}\Pr\left(\mathcal{O}\right),~\textrm{for}~f=1.
\end{equation}
Finally, combining the last equation with equation \eqref{eq_th1_case1} we get the following Theorem.

\begin{thm}
\label{thm_exp_equality}
The minimum (among all coding schemes) probability of codeword error, $P^*_E$ of the DDF protocol is exponentially equal to the probability of the event $\mathcal{O}$ defined in (\ref{outage_definition}), i.e.
\begin{equation}
\label{outage_theorem_equation}
P^*_E   ~\dot{=} ~~\Pr(\mathcal{O})
\end{equation}
\end{thm}

\subsection{SNR exponent of $\Pr(\mathcal{O})$ }

In what follows we shall refer to $\mathcal{O}$ as the outage event and $\Pr(\mathcal{O})$ as the outage probability. By definition \eqref{def_optimal_diversity_order} and Theorem~\ref{thm_exp_equality} it is clear that the the negative SNR exponent of the outage probability is equal to the optimal diversity order of the DDF protocol, i.e.,
\begin{equation}\label{eq_alternate_def_diversity_order}
    d^*(r)=\lim_{\rho \to \infty}\frac{-\log\left(\Pr(\mathcal{O})\right)}{\log(\rho)}
\end{equation}
For asymptotically high value of $\rho$, $I(H)$ can be written as
\begin{IEEEeqnarray}{rl}
\label{mutual_information}
I(H)\dot{=}&~f\log\left(\det\left(I_n+\rho H_{SD}H_{SD}^{\dagger}\right)\right)+ (1-f)\log\left(\det\left(I_n+\rho (H_{SD}H_{SD}^{\dagger}+H_{RD}H_{RD}^{\dagger})\right)\right),\nonumber\\
\dot{=}&~\log\left(\det(I_n+\rho H_{SD}H_{SD}^{\dagger})\right)+(1-f) \log\left(\det (I_n+\rho H_{RD}H_{RD}^{\dagger}
(I_n+\rho H_{SD}H_{SD}^{\dagger})^{-1} )\right)\nonumber, \nonumber\\
\label{mutual_information}
\dot{=}&~\log\left(\det(I_n+\rho H_{SD}H_{SD}^{\dagger})\right)+(1-f) \log\left(\det (I_k+\rho H_{RD}^{\dagger}
(I_n+\rho H_{SD}H_{SD}^{\dagger})^{-1} H_{RD})\right)\nonumber.
\end{IEEEeqnarray}
Note that in the above expression $f$ depends on $\bar{\gamma}$ through equation~\eqref{eq_f_dependence_on_channel}. The distribution of $\bar{\gamma}$ for asymptotic $\rho$ is given by~\cite{tse1}
\begin{equation}
\label{gamma_distribution}
\mathbf{h}(\bar{\gamma})\dot{=}\left\{\begin{array}{cc}
{\rho}^{-\sum_{i=1}^{t}(k+m-2i+1)\gamma_i},~&~\textrm{if}~\bar{\gamma}\in\mathcal{D};\\
0,&~\textrm{if}~\bar{\gamma}\notin\mathcal{D},
\end{array}\right.
\end{equation}
where $\mathcal{D}=\{ 0\leq \gamma_1\leq \gamma_2\leq \cdots \gamma_t\}$. Putting $H_2=H_{SD}$ and $H_1=H_{RD}$ in Theorem~\ref{thm_eigenvalue_distribution_2}, the above expression can be written as

\begin{eqnarray}
\label{alphabeta1}
I(H)\dot{=}\left(\sum_{i=1}^{p}(1-\alpha_i)^++(1-f)\sum_{j=1}^{q}(1-\beta_j)^+\right)\log(\rho),
\end{eqnarray}
where the joint pdf of $\bar{\alpha}$ and $\bar{\beta}$ is given by equation~(\ref{distribution2}).
Substituting this equivalent expression for $I(H)$ into the definition of the outage event we see that the outage probability $\Pr(\mathcal{O})$ depends on the different channel matrices only through the joint distribution of $\bar{\alpha},\bar{\beta}$ and $\bar{\gamma}$. Further, since $\bar{\gamma}$ is independent\footnote{ Because $\bar{\gamma}$ is a function of $H_{SR}$ whereas $(\bar{\alpha},\bar{\beta})$ is a function of $H_{RD}$ and $H_{SD}$ only, and does not depend on $H_{SR}$} of  $(\bar{\alpha},\bar{\beta})$, the outage probability can be written as
\begin{equation}
\label{outage_integration}
\Pr(\mathcal{O})=\int_{(\bar{\alpha},\bar{\beta},\bar{\gamma})\in \mathcal{O}} \mathbf{g}(\bar{\alpha},\bar{\beta})\mathbf{h}(\bar{\gamma}) \,\mathrm{d} \bar{\alpha}\,\mathrm{d} \bar{\beta}\,\mathrm{d} \bar{\gamma},
\end{equation}
where $\mathcal{O}$ is given, using (\ref{outage_definition}) and (\ref{alphabeta1}), as
\begin{IEEEeqnarray}{rl}
\label{outage_definition2}
\mathcal{O}=\Bigg\{(\bar{\alpha},\bar{\beta},\bar{\gamma}):\left(\sum_{i=1}^{p}(1-\alpha_i)^++(1-f)\sum_{j=1}^{q}(1-\beta_j)^+\right)~\leq& r;\\
0\leq \min\left\{1,\frac{r}{\sum_{l=1}^{k}(1-\gamma_l)^+}\right\}=&f;\Bigg\}.
\end{IEEEeqnarray}
Finally, evaluating the integral in equation \eqref{outage_integration} we get the following theorem.
\begin{thm}
\label{thm_optimization_problem}
The optimal diversity order, $d^*(r)$, of the DDF protocol at any multiplexing gain $r$ is given by
\begin{equation}
\label{eq_thm_opt_prob_eq1}
d^*(r)=\min \left\{\left(\hat{d}(r)+\varphi(r,t)\right),\Big(d_{m,n}(r)+d_{k,m}(r)\Big)\right\},~\textrm{for}~0\leq r\leq \min\{m,n\},
\end{equation}
where the $\varphi(\cdot,\cdot)$ function is defined as in (\ref{varphi-fn}), $d_{m,n}(r)$ represents the diversity order of a MIMO PtP channel at a multiplexing gain of $r$~\cite{tse1} and
\begin{IEEEeqnarray*}{rl}
\hat{d}(r)= &\min_{1\leq i\leq 3}~\min_{\left\{y\in \mathcal{R}_i,~b\in \mathcal{B}_i(y)\right\}} F\left(\phi_{\alpha}\left(r-b\left(1-\frac{r}{y}\right)\right), \phi_{{\beta}}(b), \phi_{{\gamma}}\left(y\right) \right),\\
\mathcal{B}_1(y)=&\left[0,\frac{y(n-r)}{r}\right];~
\mathcal{B}_2(y)=\left[0,q\right];~
\mathcal{B}_3(y)=\left[0, \frac{ry}{(y-r)}\right],\\
\mathcal{R}_1=&\left(r, \frac{qr}{(n-r)}\right];
~\mathcal{R}_2=\left(\frac{qr}{(n-r)}, \frac{qr}{(q-r)}\right];
~\mathcal{R}_3=\left(\frac{qr}{(q-r)}, t\right],
\end{IEEEeqnarray*}
with
\begin{equation}
F\left(\bar{\alpha},\bar{\beta},\bar{\gamma}\right) \stackrel{\Delta}{=} E \left(\bar{\alpha},\bar{\beta}\right) +\sum_{i=1}^{t}(k+m-2i+1)\gamma_i
\end{equation}
and the vectors $\phi_\alpha (\cdot ) $, $\phi_\beta (\cdot ) $ and $\phi_\gamma (\cdot ) $ defined in Appendix C in equations (\ref{eq_phi_alpha}), (\ref{eq_phi_beta}) and (\ref{eq_phi_gamma}), respectively.

\end{thm}

\begin{proof}[Proof~(Outline)]
The probability of outage can be expressed as
\begin{equation}
\label{eq_pf_optimization_problem_t1}
\Pr\left(\mathcal{O}\right)=\Pr\left(\mathcal{O}|f<1\right)\Pr\left(f<1\right)+\Pr\left(\mathcal{O}|f=1\right)\Pr\left(f=1\right).
\end{equation}
Using the result from \cite{tse1} in equation~\eqref{eq_outage_event_f1} we get
\begin{equation}
\label{eq_pf_optimization_problem_t2}
\Pr\left(\mathcal{O}|f=1\right)=\Pr\left\{\sum_{i=1}^{p}(1-\alpha_i)^+\leq r\right\}\dot{=}\rho^{-d_{m,n}(r)}.
\end{equation}
From the definition of $f$ in equation \eqref{ratio} and equation \eqref{gamma_distribution}, we get
\begin{equation}
\label{eq_pf_optimization_problem_t3}
\Pr\left(f=1\right)=\Pr\left\{\sum_{l=1}^{t}(1-\gamma_l)^+\leq r\right\}\dot{=}\rho^{-d_{m,k}(r)}.
\end{equation}
Now, using the fact that $f\in [0,1]$ with equation \eqref{eq_pf_optimization_problem_t3} we get
\begin{equation}
\label{eq_pf_optimization_problem_t4}
\Pr\left(f<1\right)=1-\Pr\left(f=1\right)\dot{=}1-\rho^{-d_{m,k}(r)}\dot{=}\rho^{-\varphi(r,t)},
\end{equation}
because for $r\geq t,~d_{m,k}(r)=0$. Finally, denoting the negative SNR exponent of $\Pr\left(\mathcal{O}|f<1\right)$ by $\hat{d}(r)$, i.e., $\Pr\left(\mathcal{O}|f<1\right)\dot{=}\rho^{-\hat{d}(r)}$, and combining it with equations \eqref{eq_pf_optimization_problem_t1}, \eqref{eq_pf_optimization_problem_t2}, \eqref{eq_pf_optimization_problem_t3} and \eqref{eq_pf_optimization_problem_t4} we get
\begin{equation*}
\Pr\left(\mathcal{O}\right)\dot{=}\rho^{-\Big(\hat{d}(r)+\varphi(r,t)\Big)}+\rho^{-\Big(d_{m,n}(r)+d_{k,m}(r)\Big)},
\end{equation*}
which imply
\begin{equation*}
d^*(r)=\min \left\{\left(\hat{d}(r)+\varphi(r,t)\right),\Big(d_{m,n}(r)+d_{k,m}(r)\Big)\right\}.
\end{equation*}
To complete the proof it is only necessary to compute $\hat{d}(r)$, for which we need to evaluate the integral in equation \eqref{outage_integration} under the constraint $f<1$. This integral can be evaluated using Laplace's method of integration as in \cite{tse1} to get $\hat{d}(r)$ as the minimum value of the negative SNR exponent of the pdf of $(\bar{\alpha},\bar{\beta},\bar{\gamma})$, minimized over the intersection of the outage set and the support set of the pdf. Evaluating this minimum value directly is not prescribed for two reasons: 1) it is not a standard convex optimization problem; and 2) the number of optimizing variables increase linearly with the number of antennas at all the nodes (i.e., with $(m+k+n)$). To overcome these problems we first transform the original minimization problem into an equivalent optimization problem having only three variables and then, analyzing it further, we eventually get the much simpler convex optimization problem given in the theorem statement involving only two variables.

This is done in Appendix~\ref{pf_thm_optimization_problem}.
\end{proof}

Note that the main step in computing the DMT of the DDF protocol for any given antenna configuration is the computation of $\hat{d}(r)$. We illustrate this step by an example.

\begin{ex}[DMT on the $(1,1,1)$ channel]
\label{ex_dmt_111_ch}
Putting $n=m=k=1$ in the expression for $F$ we get
\begin{equation*}
G(a,b,y)=F\left(\phi_{\alpha}(a),\phi_{\beta}(b),\phi_{\gamma}(y)\right)=2(1-r)+2b\left(1-\frac{r}{y}\right)-b+(1-y).
\end{equation*}
From Theorem~\ref{thm_optimization_problem} we know that the above objective function has to be minimized over three different set of $(a,b)$ pairs depending on the value of $y$.
Since
\begin{equation*}
\left(1-\frac{r}{y}\right)\leq \frac{1}{2}~\forall~y~\in \mathcal{R}_1,
\end{equation*}
the objective function attains its minimum at the maximum value of $b$ in $\mathcal{B}(y)$, i.e., $b^*=y\left(\frac{1-r}{r}\right)$. Putting this into the objective function we have
\begin{equation}
\label{eq_ex1_temp1}
G(a^*,b^*,y)=y\left(\frac{1-r}{r}\right)+(1-y)=1+y\left(\frac{1-2r}{r}\right).
\end{equation}
It is clear that the optimal value of $y\in \mathcal{R}_1=\left(r,\frac{r}{1-r}\right]$ that minimizes the above function is given as
\begin{equation*}
y^*=\left\{\begin{array}{cc}
r,&\textrm{for}~r\leq \frac{1}{2};\\
1,&\textrm{for}~r\geq \frac{1}{2}.
\end{array}\right.
\end{equation*}
Putting this solution in equation \eqref{eq_ex1_temp1} we get
\begin{equation}
\label{eq_dmt_111_R1}
d\left(\mathcal{R}_1\right)=\left\{\begin{array}{cc}
2(1-r),&\textrm{for}~0\leq r\leq \frac{1}{2};\\
\left(\frac{1-r}{r}\right),&\textrm{for}~\frac{1}{2}\leq r\leq 1.
\end{array}\right.
\end{equation}
Since for $m=n=k$, $\mathcal{R}_2=\Phi$, next we consider the case when $y\in \mathcal{R}_3=\left(\frac{r}{1-r},1\right]$. This set is non-empty only for $r\leq \frac{1}{2}$ and the optimal point lies on the line segment DF in Figure~\ref{y_range-c}. Dividing the set $\mathcal{R}_3$ further into two subsets $\mathcal{R}_{31}=(r,2r]$ and $\mathcal{R}_{32}=(2r,1]$ we see that $\left(1-\frac{r}{y}\right)\leq \frac{1}{2}$ when $y\in \mathcal{R}_{31}$ and $\left(1-\frac{r}{y}\right)\geq \frac{1}{2}$ when $y\in \mathcal{R}_{32}$. The objective function attains its minimum value at point F, where $b$ is maximum when $y\in \mathcal{R}_{31}$ and at point D, where $b$ is minimum when $y\in \mathcal{R}_{32}$ and given by
\begin{equation*}
G(a^*,b^*,y)=\left\{\begin{array}{cc}
3-\frac{yr}{y-r}-y,&\textrm{for}~y\in \mathcal{R}_{31};~\left[b^*=\frac{yr}{y-r}\right]\\
2(1-r)+(1-y),&\textrm{for}~\textrm{for}~y\in \mathcal{R}_{32},~\left[b^*=0\right]
\end{array}\right.
\end{equation*}
Now, optimizing this function in the corresponding sets of $y$ we get
\begin{equation*}
G(a^*,b^*,y^*)=\left\{\begin{array}{cc}
3-4r,&\textrm{for}~y\in \mathcal{R}_{31};~\left[\because ~y^*=2r\right]\\
2(1-r),&\textrm{for}~\textrm{for}~y\in \mathcal{R}_{31};~\left[\because ~y^*=1\right]
\end{array}\right.
\end{equation*}
which in turn implies
\begin{equation}
\label{eq_dmt_111_R3}
d\left(\mathcal{R}_3\right)=2(1-r),~\textrm{for}~0\leq r\leq \frac{1}{2}.
\end{equation}
Combining equations \eqref{eq_dmt_111_R1} and \eqref{eq_dmt_111_R3} we get
\begin{equation*}
\hat{d}(r)=\min \left\{d\left(\mathcal{R}_0\right),d\left(\mathcal{R}_1\right),d\left(\mathcal{R}_3\right)\right\}=\left\{\begin{array}{cc}
2(1-r),&\textrm{for}~0\leq r\leq \frac{1}{2};\\
\left(\frac{1-r}{r}\right),&\textrm{for}~\frac{1}{2}\leq r\leq 1.
\end{array}\right.
\end{equation*}
Using this in Theorem~\ref{thm_optimization_problem} we see that the optimal diversity order of the DDF protocol on a $(1,1,1)$ channel is given by $\hat{d}(r)$, thereby recovering the result of \cite{KHP}.
\end{ex}

\section{Closed form expressions for the DMT of a few simple relay channels}
\label{sec_closed_form_dmt_calculations}
In this section, we shall provide closed form expressions for the DMT of the DDF protocol for three more general channel configurations (than the previous example), namely, the $(n,1,n)$ channel, the $(1,k,1)$ channel (for $k\geq 1$) and the $(2,k,2)$ channel (for $k\geq 2$) by solving the optimization problem in Theorem \ref{thm_optimization_problem}.

\subsection{DMT of the $(n,1,n)$ channel}
\begin{thm}
\label{thm_analytical1}
The optimal DMT of the DDF protocol on a $(n,1,n)$ half-duplex relay channel for multiplexing gains $r$ is
\begin{equation}
\label{eq_dmt_n1n_rlessthan1}
d^*_1(r)=\left\{\begin{array}{cc}
\frac{1-r}{\max\{\frac{1}{2},r\}},&~0\leq r\leq 1~\textrm{and}~ n=1;\\
(n-1)^2+(3n-1)(1-r),&~0\leq r\leq 1~\textrm{and}~ n\geq 2;\\
d_{n,n}(r),&~1\leq r\leq n~\textrm{and} ~n\geq 2.
\end{array}\right.
\end{equation}
\end{thm}
\begin{proof}
The proof is given in Appendix~\ref{pf_thm_analytical1}.
\end{proof}

\begin{remark}
The optimal DMT of the full-duplex decode-and-forward (FD-DF) protocol was derived in~\cite{YEE} and over a $(n,1,n)$ relay channel it is given as
\begin{equation*}
d_{DF}(r)=\left\{\begin{array}{c}\min\{d_{(n+1),n}(r), d_{n,n}(r)+d_{n,1}(r)\}, 0\leq r\leq 1;\\
d_{n,n}(r), ~1\leq r\leq n,\end{array}\right.
\end{equation*}
where $d_{n,m}(r)$ represents the diversity of a $n$-transmit, $m$-receive antenna MIMO channels diversity order at multiplexing gain of $r$. For $n\geq 2$, the DMT given by the above equation is {\em identical} with that given in Theorem~\ref{thm_analytical1}. Thus, when the relay has a single antenna and the source and destination have $n$ antennas each, the optimal DMT of the half-duplex DDF protocol and a FD-DF protocol are identical.
\end{remark}

\subsection{DMT of the $(1,k,1)$ channel}

\begin{thm}
\label{thm_analytical2}
The optimal DMT of the DDF protocol on a $(1,k,1)$ half-duplex relay channel is
\begin{equation}
d_2^*(r)~\triangleq \left\{\begin{array}{c}
(k+1)(1-r), ~~0\leq r\leq \frac{1}{k+1};\\
1+k(\frac{1-2r}{1-r}), ~\frac{1}{k+1}\leq r \leq \frac{1}{2};\\
\left(\frac{1-r}{r}\right), ~~\frac{1}{2}\leq r \leq 1.
\end{array}\right.
\end{equation}
\end{thm}

\begin{proof}
The proof is given in Appendix~\ref{pf_thm_analytical2}.
\end{proof}

\begin{remark}
For $k\geq 2$, the optimal DMT of the SCF protocol on a $(1,k,1)$ channel is given by~\cite{OCM} as the piece-wise linear curve whose values at three corner points are given as $d_{SCF}(0)=(1+k), ~~d_{SCF}(\frac{1}{2})=1$ and $d_{SCF}(1)=0$. Comparing this with the corresponding DMT of the DDF protocol given above we see for $0\leq r\leq \frac{1}{2}$, the DDF protocol can achieve better diversity order (see Figure~\ref{DDF_SCF_comparison1}) while requiring less channel state information. Moreover, since the SCF protocol achieves the best DMT among all static protocols, it is evident (for example, see Figure~\ref{DDF_SCF_comparison1}) that the DDF protocol can perform better than the SCF protocol, that a static protocol is not DMT optimal on a MIMO HD relay channel.
\end{remark}

\begin{figure}[!hbt]
\centering
\includegraphics[width=8.0cm,height=6cm,]{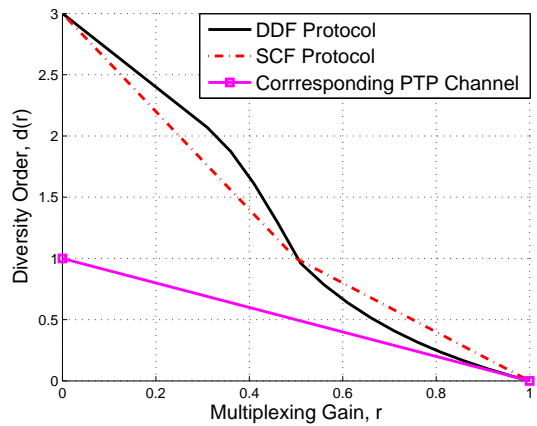}
\caption{DMT comparison of the DDF and SCF protocol on a $(1,2,1)$ relay channel.}
\label{DDF_SCF_comparison1}
\end{figure}

\begin{remark}
Recently, the fundamental DMT of the $(1,k,1)$ relay channel was derived by the authors in \cite{Sanjay_Varanasi_Symmetric_HDRC_DMT}
which is given as
\begin{equation*}
d_{(1,k,1)}(r)~\triangleq \left\{\begin{array}{c}
(k+1)(1-r), ~~0\leq r\leq \frac{1}{k+1};\\
1+k(\frac{1-2r}{1-r}), ~\frac{1}{k+1}\leq r \leq \frac{1}{2};\\
2\left(1-r\right), ~~\frac{1}{2}\leq r \leq 1.
\end{array}\right.
\end{equation*}
Comparing it with the result of Theorem~\ref{thm_analytical2} we see that the DDF protocol can achieve the fundamental DMT of the channel for multiplexing gains in the range $0\leq r\leq \frac{1}{2}$. However, DMT optimality of the DDF protocol is not restricted to just this channel. In the next section we shall see that the DDF protocol can achieve the fundamental DMT of the channel for other antenna configurations also for some multiplexing gains.
\end{remark}

\begin{remark}
Note if Theorem~\ref{thm_analytical1} and Theorem~\ref{thm_analytical2} are specialized to the cases of $n=1$ and $k=1$, respectively, one recovers the result derived in \cite{KHP}.
\end{remark}

\subsection{DMT of the $(2,k,2)$ channel}
\begin{thm}
\label{analytical3}
An upper bound to the optimal DMT of the DDF protocol on a $(2,k,2)$ relay channel is given by the following
\begin{equation*}
d_u(r)=\min\left\{\begin{array}{c}
d_{2,2}(r)+d_{2,k}(r), ~0\leq r\leq 2;\\
k+3+(k+1)\left(\frac{2-3r}{2-r}\right),~0\leq r\leq \frac{2}{3};\\
k+6\left(\frac{1-r}{r}\right),~\frac{2}{3}\leq r\leq 1;\\
4+4(k-1)\left(\frac{1-r}{2-r}\right),~\frac{2}{3}\leq r\leq 1;\\
1+(k-1)\left(\frac{4-3r}{2-r}\right),~1\leq r \leq \frac{4}{3};\\
4\left(\frac{3-2r}{r}\right),~1\leq r \leq \frac{4}{3};\\
2\left(\frac{2-r}{r}\right),~\frac{4}{3}\leq r \leq 2;
\end{array}\right.
\end{equation*}
\end{thm}
\begin{proof}
The proof is given in Appendix~\ref{pf_thm_analytical3}.
\end{proof}
\begin{rem}
The DMT of the DDF protocol on $(2,k,2)$ channel was also computed using the numerical method (described later in this section) for the general $(m,k,n)$ channel. It was observed that the DMT coincides with the upper bound given by the above theorem for $k\leq 20$. Thus the upper bound of Theorem \ref{analytical3} is tight for $k\leq 20$. Given the tightness of the bound for $k\leq 20$ we conjecture is that the upper bound of Theorem \ref{analytical3} is tight for all $k\geq 2$.
\end{rem}

\begin{figure}[!hbt]
\centering
\includegraphics[width=8.0cm,height=6cm,]{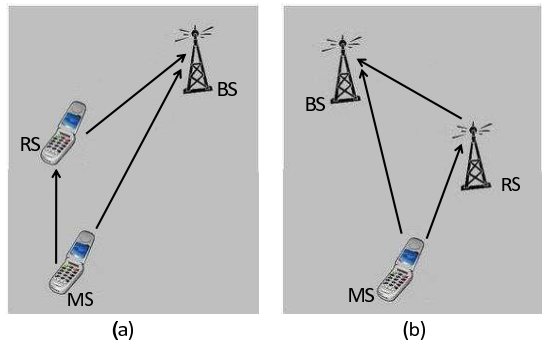}
\caption{Cooperative networks: (a) $\textrm{CN}_1$: A mobile station act as a relay node; (b) $\textrm{CN}_2$: A dedicated relay station acts as the relay node.}
\label{cooperative_scenario}
\end{figure}

\begin{figure}[htp]
  \begin{center}
    \subfigure[Uplink Channel]{\label{scenario1-a}\includegraphics[scale=0.5]{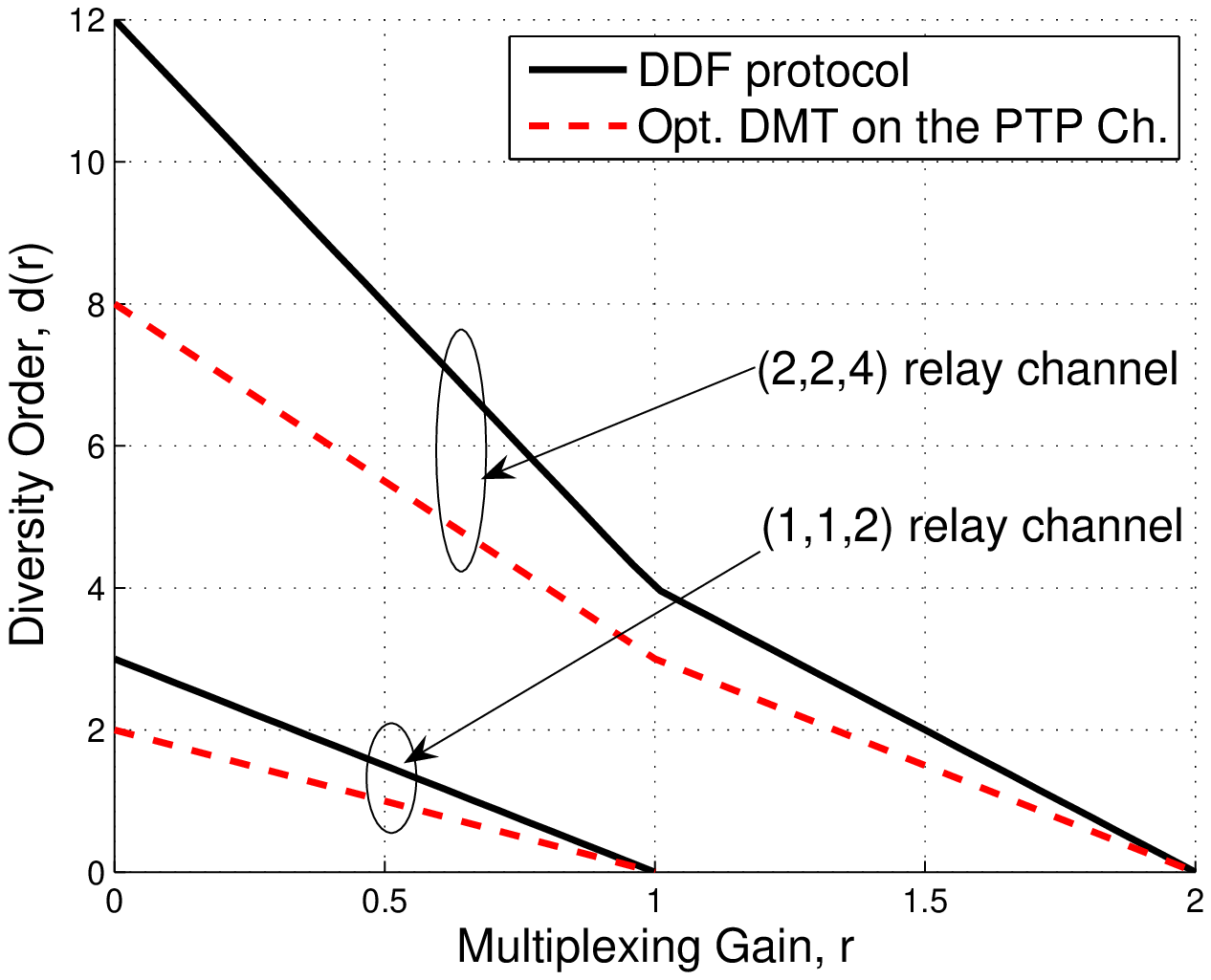}}
    \subfigure[Downlink Channel]{\label{scenario1-b}\includegraphics[scale=0.5]{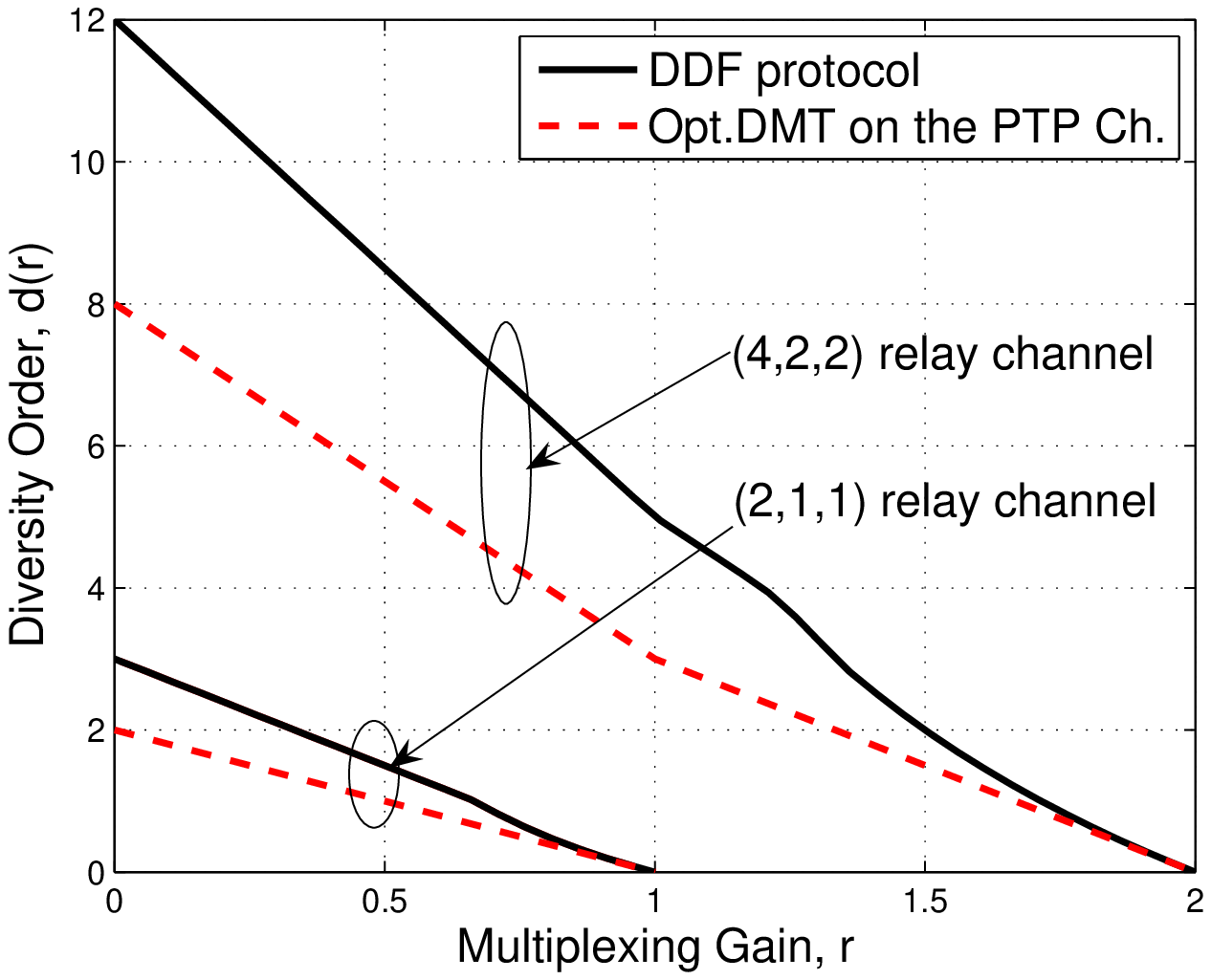}}
  \end{center}
  \caption{DMT curves of the DDF protocol on relay channels of $\textrm{CN}_1$.}
  \label{scenario1}
\end{figure}

%
%

\begin{figure}[htp]
  \begin{center}
    \subfigure[Uplink Channel]{\label{scenario2-a}\includegraphics[scale=0.5]{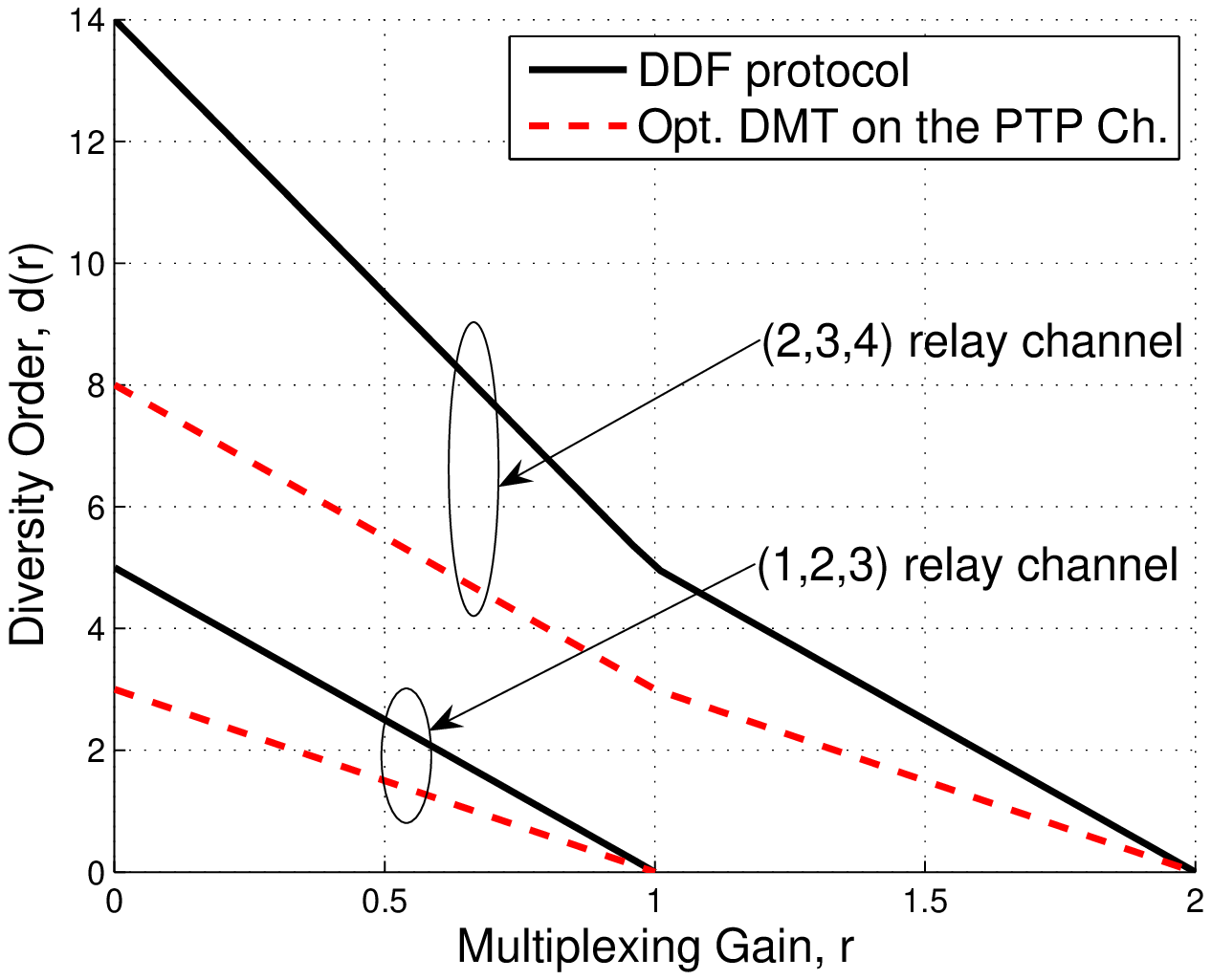}}
    \subfigure[Downlink Channel]{\label{scenario2-b}\includegraphics[scale=0.5]{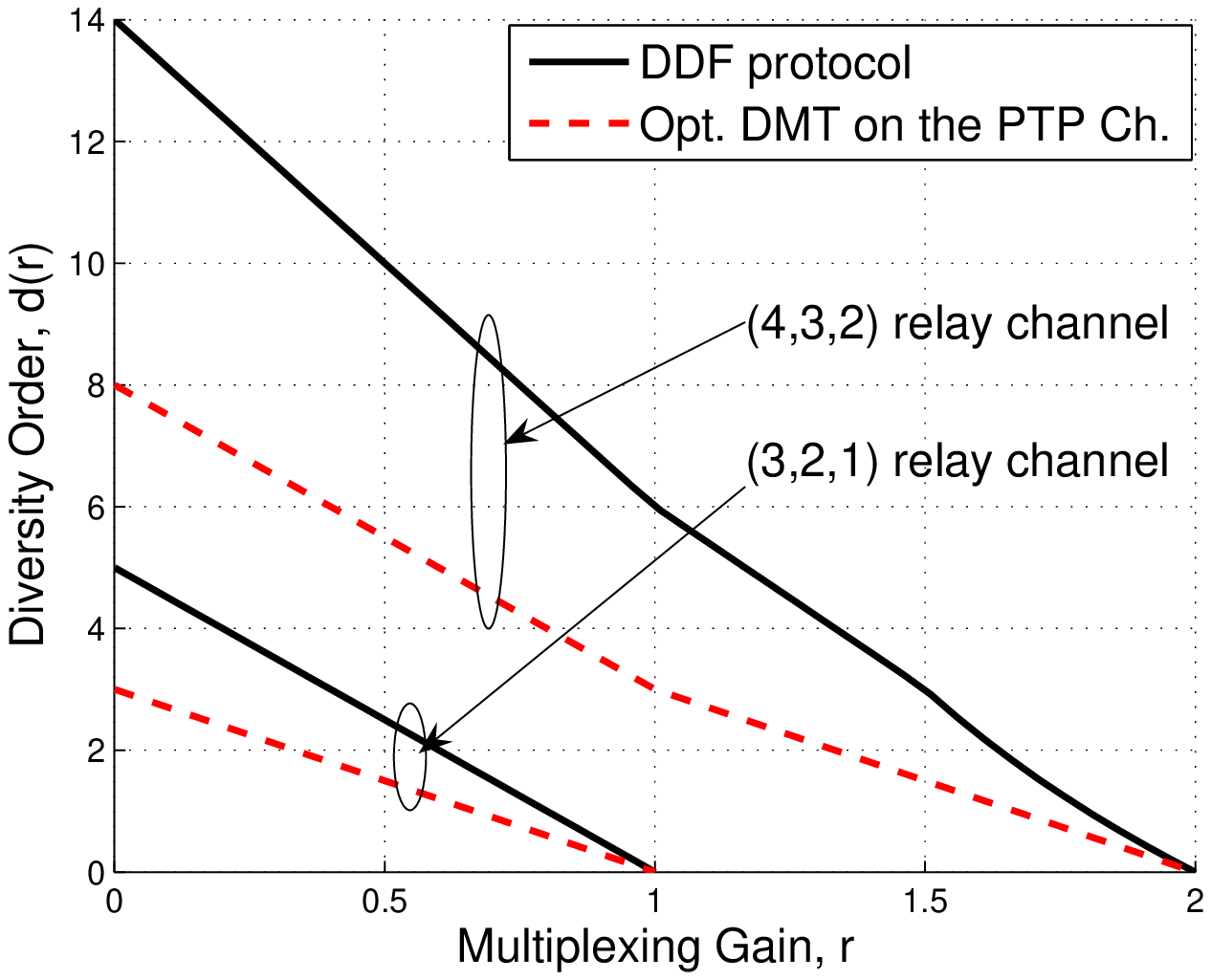}}
  \end{center}
  \caption{DMT curves of the DDF protocol on relay channels of $\textrm{CN}_2$.}
  \label{scenario2}
\end{figure}

\section{Explicit DMTs of the DDF protocol and a comparative study in practical applications}
\label{sec_explicit_dmt}

In this subsection, we illustrate the advantage of the DDF protocol over PtP communication schemes and provide a comparative analysis of its performance with respect to other MIMO cooperative schemes. We consider different practical scenarios where cooperative communication promises potential gain in overall system performance because of which it is being considered by several standardization bodies. Figure~\ref{cooperative_scenario}a depicts a cellular network wherein, a mobile user (or mobile set (MS)) uses another mobile user as the relay station (RS) to communicate its message to and from the base station (BS). This cooperative model was first proposed in \cite{SEB1}. Figures~\ref{scenario1-a} and \ref{scenario1-b} represent the uplink and down-link performances, respectively, of the DDF protocol in such an environment with respect to the fundamental DMT of the corresponding $m$-transmit, $n$-receive antennas PtP channel.

Figure~\ref{cooperative_scenario}b depicts a scenario where, in a cellular network (CN), a particular cell area is divided into more than one sub-cell and each sub-cell is served by an additional dedicated node (a smaller BS) to provide better quality of service. Thus each user in these sub-cells can use these dedicated nodes to relay their messages to and from the BS. This is different from the previous cooperative scenario in the sense that this relay stations can host more number of antennas than a mobile set. This configuration is under consideration to be implemented in LTE-advanced and WiMAX technologies~\cite{YHXM} and being standardized by the IEEE 802.16s relay task group~\cite{relayTG}. Figures~\ref{scenario2-a} and \ref{scenario2-b} show the uplink and downlink performances respectively, of the DDF protocol in such a scenario. These figures clearly demonstrate the superior performance of the DDF protocol over that of the corresponding MIMO channel.

Note the asymmetry in terms of number of antennas at different nodes in both of the above applications, which points out the importance of analyzing MIMO relay channels with an arbitrary number of antennas at each node.

\begin{figure}[htp]
  \begin{center}
    \subfigure[$\textrm{CN}_3$: Sensory network with a mobile relay station (MRS)]{\label{scenario3-a}\includegraphics[scale=0.27]{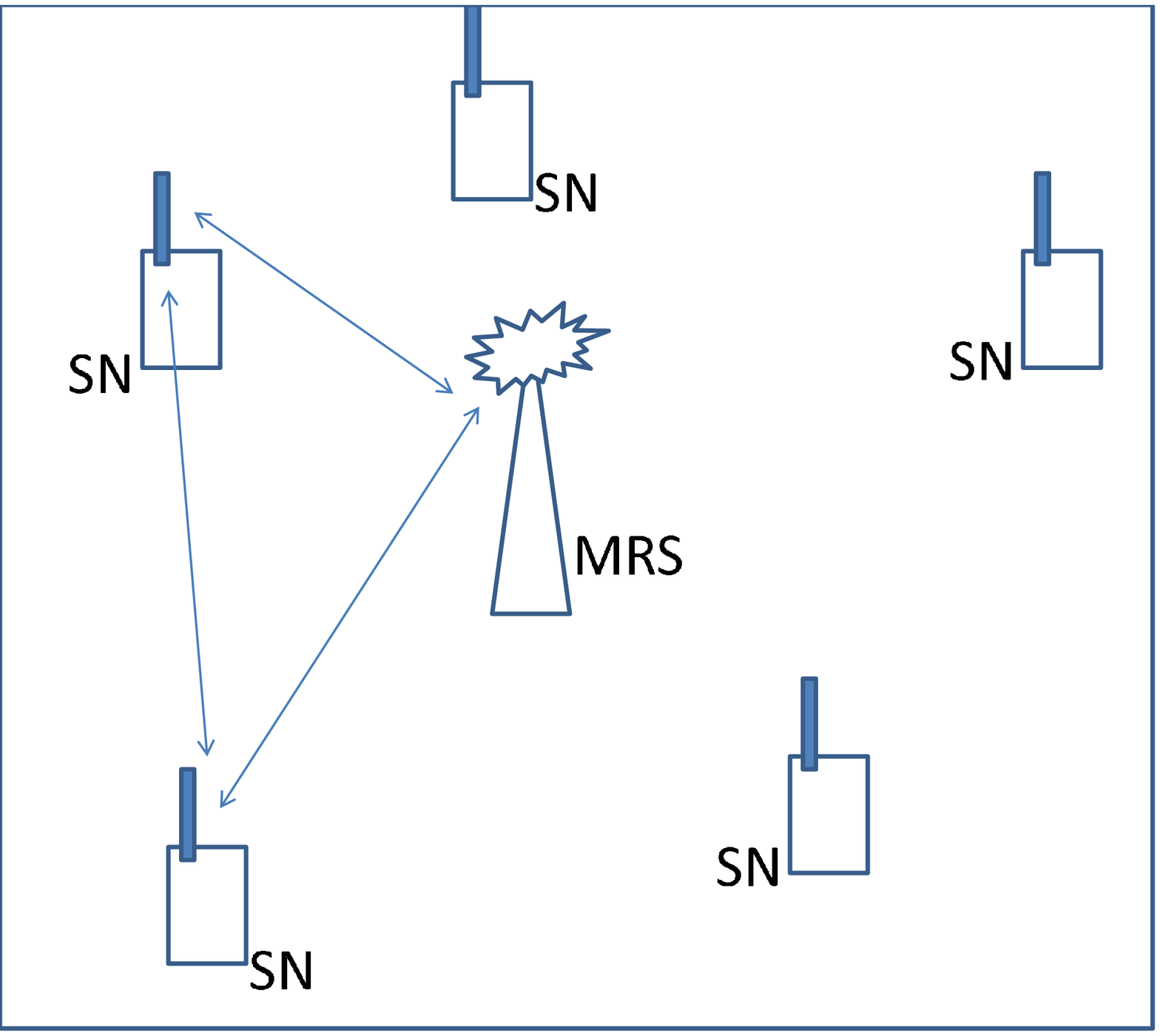}}
    \subfigure[DMT comparison of DDF and SCF protocol]{\label{scenario3-b}\includegraphics[scale=0.5]{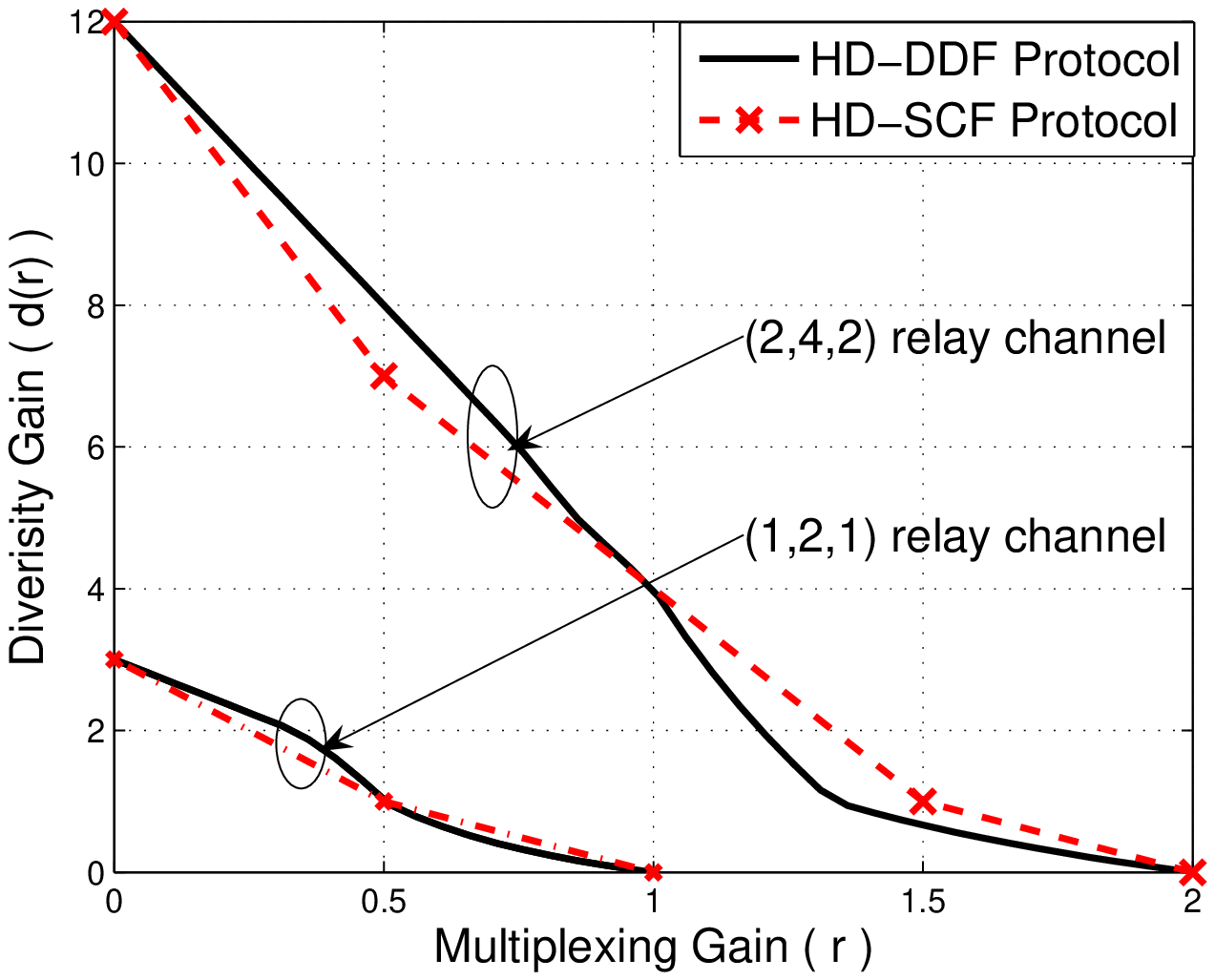}}
  \end{center}
  \caption{Explicit DMTs of the DDF and SCF protocols on relay channels of $\textrm{CN}_3$.}
  \label{scenario3}
\end{figure}

Among the many cooperative protocols available in the literature, the static CF and full-duplex DF are the two protocols that are known to have high performance and whose explicit DMT characterizations on a MIMO relay channel are now known. The DMTs of the DDF and the SCF protocols on the sensor network of Figure~\ref{scenario3-a} first proposed in \cite{WSC} where a more capable mobile relay station (with more antennas) helps several less capable sensor nodes to communicate with each other, is plotted in Figure~\ref{scenario3-b}. This figure illustrates that on such a relay channel, the DDF protocol can outperform the SCF protocol at low multiplexing gains. However, on a relay channel of some other CN, such as the downlink of $\textrm{CN}_1$, the SCF protocol performs marginally better than or identical to that of the DDF protocol uniformly at all multiplexing gains, as shown in Figure~\ref{protocol_comparison_scenario1-a}.

The implementation of a protocol in a practical application however, depends on a number of other issues, with an important one being the channel state information (CSI) required at different nodes. The SCF protocol, in contrast to the DDF protocol, requires global CSI\footnote{ The knowledge of all the instantaneous channel matrices of the relay channel.} at the relay node, which in some application-such as the downlink channel of $\textrm{CN}_1$- may be a challenging or even impossible task. The CSI assumptions for the DDF protocol are that the receivers know the incoming channels which is a much milder assumption given this can be accomplished via training. It is also assumed in this work that the destination has perfect knowledge of the relay decision time which is a function of the source-relay channel. Future work on the DDF protocol for MIMO relay channels in the spirit of bringing theory closer to practice would be to relax the assumptions of infinite block length and genie-aided relay decision time information at the destination as was done for the single-antenna relay channel in \cite{RC}, for which the results of this paper would serve as a benchmark. The DDF protocol would thus provide a practical alternative to the SCF protocol without sacrificing much by way of performance. For instance, on the $(4,2,2)$ relay channel of Figure~\ref{protocol_comparison_scenario1-a}, the burden of providing global CSI to the relay node which is an MS having limited capability can be avoided through negligible loss in diversity order.

The performance comparison of the DDF protocol with that of the SCF protocol depicted in Figure~\ref{scenario3-b} is also interesting in light of a recent result~\cite{PAT}, where it was proved that, on a {\em single-antenna} relay channel, a static protocol, namely the quantize-and-map protocol, is DMT optimal among all static and dynamic cooperative protocols. This result~\cite{PAT} raises a natural question: can a static protocol achieve the fundamental DMT of a MIMO half-duplex relay channel? The analysis of this paper shows that the DDF protocol can be better than the theoretical limit of static protocols, which is the DMT of the SCF protocol, and thus answers the above question in the negative. The comparison in Figure~\ref{scenario3-b} proves that static protocols fundamentally can not fully exploit the resources available on a HD MIMO relay channel.

In Figure~\ref{protocol_comparison_scenario1-b}, we compare the performance of the DDF protocol with the full-duplex DF protocol on the uplink channel of Figure~\ref{cooperative_scenario}(a). This figure illustrates that the dynamic operation of the half-duplex relay node in the DDF protocol can help achieve almost the same or equal performance as in the FD-DF protocol without full-duplex relaying. While the large difference between the transmitted and received power levels makes full-duplex operation impractical, if not entirely infeasible, the DDF protocol may be implemented with little or no loss in performance while avoiding the cost of full-duplex operation.

\begin{figure}[htp]
  \begin{center}
    \subfigure[Downlink Channel]{\label{protocol_comparison_scenario1-a}\includegraphics[scale=0.5]{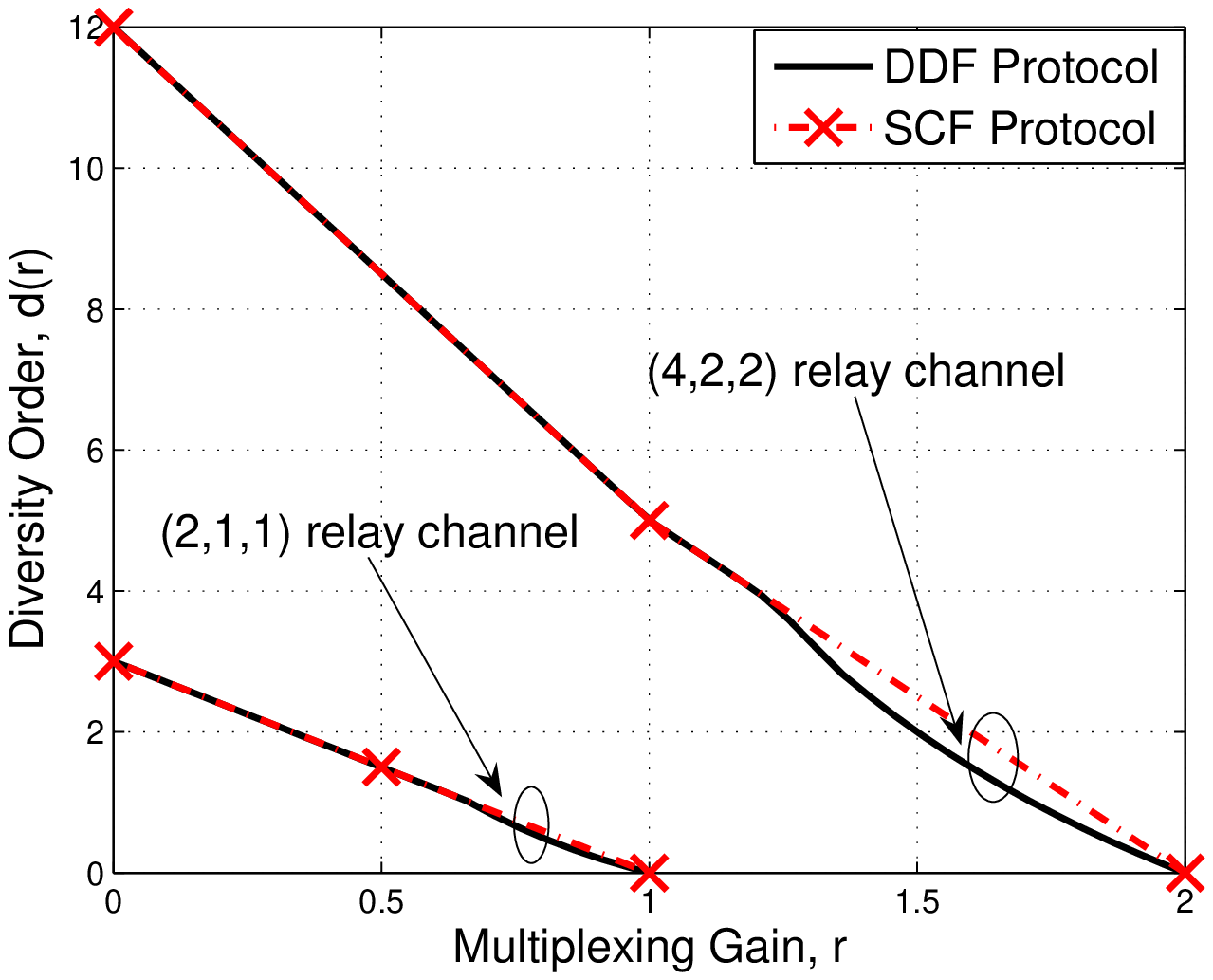}}
    \subfigure[Uplink Channel]{\label{protocol_comparison_scenario1-b}\includegraphics[scale=0.5]{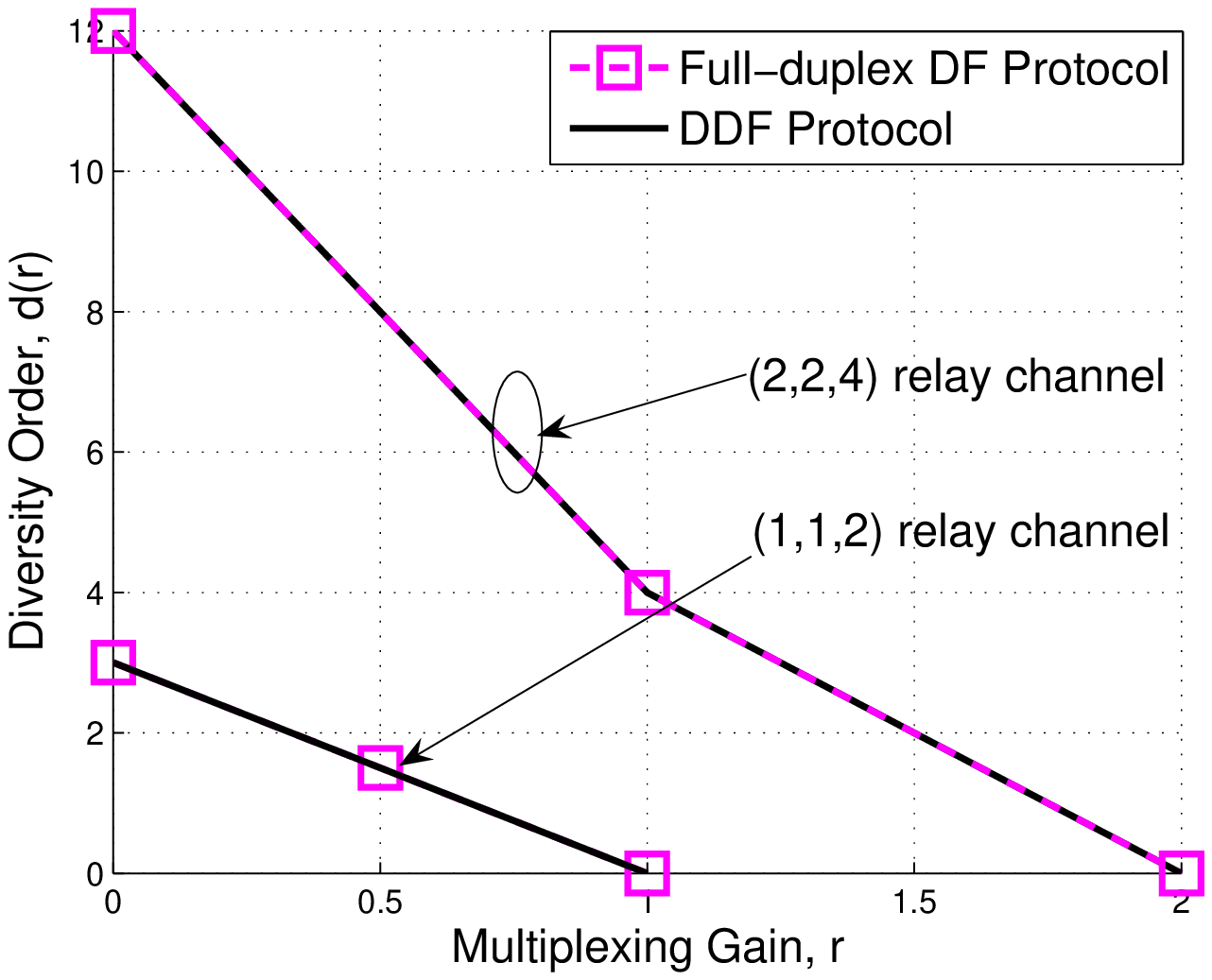}}
  \end{center}
  \caption{DMT comparison among the DDF, SCF and the FD-DF protocols on relay channels of $\textrm{CN}_1$.}
  \label{protocol_comparison_scenario1}
\end{figure}


%

\section{Conclusion}
\label{conclusion}
The asymptotic joint eigenvalue distribution of two specially correlated random Wishart matrices was derived, and using this result, the optimal diversity-multiplexing tradeoff was obtained of a three node half-duplex MIMO relay network, where each node has an arbitrary number of antennas, operating in the dynamic decode-and-forward protocol. For several specific channel configurations we computed explicit DMT curves for the protocol and compared it with FD-DF and HD-SCF protocols. These comparisons reveal some interesting facts such as, for some channel configurations, the optimal DMT of HD-DDF and FD-DF protocols are identical while for other channel configurations the diversity orders of the HD-DDF protocol are marginally less than those over FD-DF protocol at high multiplexing gains. 
Further, the comparison with the HD-SCF protocol shows that for some channel configurations, at low multiplexing gains, the optimal diversity orders of the HD-DDF protocol are greater than those of the corresponding DMT of HD-SCF protocol, which should motivate one to further investigate other dynamic protocols such as the dynamic compress-and-forward protocol, on a three node relay channel. This work also motivates further research on the generalization of the single-antenna relay channel results in \cite{RC} to the MIMO relay channel by considering finite lengths codes and doing away with the assumption of genie-aided information about the relay decision time at the destination. Extending the present analysis for a relay network having multiple relay nodes is another topic for future research.


\appendices

\section{Proof of Theorem~\ref{thm_eigenvalue_distribution_1}}
\label{pf_thm_eigenvalue_distribution_1}

We begin by proving the theorem for $N_2 \leq N_3$ and later extend the proof for $N_2 > N_3$.
Let $\lambda_1 > \lambda_2>\cdots > \lambda_{N_2}>0$ represent the ordered non-zero\footnote{ The eigenvalues of $V_2$ are distinct and non-zero with probability 1.} eigenvalues of $V_2\triangleq H_2H_2^{\dagger}$. The spectral decomposition of $V_2$ can be written as $V_2 = U \Lambda U^{\dagger}$, where $\Lambda\triangleq\textrm{diag}(\bar{\lambda})$ and $U\in \mathbb{C}^{N_2\times N_2}$ is an unitary matrix containing the eigenvectors of $V_2$. We can now write $V_1\triangleq \hat{H}_1^{\dagger}\Sigma \hat{H}_1 $, where $\hat{H}_1=U^\dagger H_1$ and $\Sigma=(I_{N_2}+ \rho\Lambda )^{-1}$.

\begin{rem}
Note that since $U, \Lambda$ (e.g., see Lemma~2.6 in \cite{Tulino_Verdu}) and $H_1$ are mutually independent so are $\Lambda$ and $\hat{H}_1$. Also, since $H_1$ is unitarily invariant, $H_1$ and $\hat{H}_1$ are identically distributed.
\end{rem}

Before proceeding further let us recall a result derived in~\cite{GaS} which deals with the eigenvalue distribution of a random matrix of the similar form.
\begin{lemma}[\cite{GaS}]
\label{lem_gas}
Let $x_1> x_2> \cdots >x_M$ be the ordered non-zero eigenvalues of $Z$, where $Z=X L X^{\dagger}$ and $X\in \mathbb{C}^{M\times N}, M<N$, has mutually independent complex Gaussian vectors as its columns with zero mean and covariance $I_M$, and $y_1> y_2> \cdots >y_N$ be the ordered non-zero eigenvalues of $L$, then the density function of $Z$ is given by
\begin{equation}
f(Z)=C_1 \pi^{\frac{M(M-1)}{2}}\frac{\triangle_1\left(\bar{y},\bar{x}\right)}{|\mathbf{V}_1(\bar{x})| |\mathbf{V}_1(\bar{y})|},
\end{equation}
where $\mathbf{V}_1(\bar{x})$ and $\mathbf{V}_1(\bar{y})$ are Vandermonde matrices formed by the vectors $\bar{x}$ and $\bar{y}$, respectively and
\begin{eqnarray}
\label{eq_del1}
\triangle_1\left(\bar{y},\bar{x}\right)=\left|\begin{array}{cccccc}
1 & y_N & \cdots y_N^{N-M-1} & y_N^{N-M-1}e^{-\frac{x_M}{y_N}} & \cdots & y_1^{N-M-1}e^{-\frac{x_1}{ y_N}}\\
1 & y_{(N-1)} & \cdots y_{(N-1)}^{N-M-1} & y_{(N-1)}^{N-M-1}e^{-\frac{x_M}{ y_{(N-1)}}} & \cdots & y_2^{N-M-1}e^{-\frac{x_1}{ y_{(N-1)}}}\\
&&& \vdots &&\\
1 & y_1 & \cdots y_1^{N-M-1} & y_1^{N-M-1}e^{-\frac{x_M}{ y_1}} & \cdots & y_N^{N-M-1}e^{-\frac{x_1}{ y_1}}\\
\end{array}\right|.
\end{eqnarray}
\end{lemma}
Note that $Z$ in the above lemma has the same structure as $V_1$ but is only valid for $M<N$. Thus assuming $N_1<N_2$ we can substitute $X=H_1^{\dagger}$ and $L=\Sigma$ in Lemma~\ref{lem_gas} to obtain
\begin{equation}
\label{eq_conditional_temp1}
f(V_1|\bar{\sigma})=C_1 \frac{\triangle_1\left(\bar{\sigma},\bar{\xi}\right)}{|\mathbf{V}_1(\bar{\xi})| |\mathbf{V}_1(\bar{\sigma})|},
\end{equation}
where $C_1$ is a constant independent of all the eigenvalues and $\bar{\sigma}$ is the vector of the ordered eigenvalues of $\Sigma$, i.e.,
\begin{equation*}
    \sigma_i=\left(1+\rho \lambda_{(N_2-i+1)}\right)^{-1},~\textrm{for}~1\leq i\leq N_2.
\end{equation*}
The conditioning on $\bar{\sigma}$ in equation \eqref{eq_conditional_temp1} is present due to the randomness of $\Sigma$. However, we are interested in the distribution of the eigenvalues of $V_1$ and not $V_1$ itself.

Let us assume that the spectral decomposition of $V_1$ is given as $V_1=V\Pi V^{\dagger}$, where $V\in \mathbb{C}^{N_1\times N_1}$ is an unitary matrix and $\Pi=\textrm{diag}(x_1, \cdots x_{N_1})$ is the diagonal matrix containing the ordered eigenvalues of $V_1$. It is well known that the Jacobian of the transformation $V_1 \mapsto (\Pi, V)$ is given by
\begin{equation}
J(\Pi, V)=|\mathbf{V}_1(\bar{\xi})|^2=\prod_{i<j}^{N_1}(\xi_i-\xi_j)^2
\end{equation}
Using this expression for the Jacobian and equation \eqref{eq_conditional_temp1} we can find the conditional joint pdf of $(\Pi, V)$ conditioned on $\bar{\sigma}$. Then integrating the resulting pdf over the space of unitary matrices we get the joint distribution of the eigenvalues as
\begin{equation}
\label{thm1_pf_eqn3}
\mathbf{f_2}(\bar{\xi}|\bar{\sigma})=C_2 |\mathbf{V}_1(\bar{\xi})|^2 \frac{\triangle_1\left(\bar{\sigma},\bar{\xi}\right)}{|\mathbf{V}_1(\bar{\xi})| |\mathbf{V}_1(\bar{\sigma})|}.
\end{equation}
Evaluating this expression in general is complicated due to the term $\triangle_1\left(\bar{\sigma},\bar{\xi}\right)$ in the above expression. The following Lemma helps us to evaluate it in the high SNR regime.
\begin{lemma}
\label{lemma_eig1_help1}
If the eigenvalues $\{\xi_1, \cdots \xi_{N_1}\}$ and $\{\lambda_1, \cdots \lambda_{N_2}\}$ vary exponentially with $\rho$, then for asymptotically large $\rho$ the following identity holds
\begin{equation*}
J_1\triangleq \frac{\triangle_1\left(\bar{\sigma},\bar{\xi}\right)}{|\mathbf{V}_1(\bar{\xi})| |\mathbf{V}_1(\bar{\sigma})|}\dot{=}  \prod_{i=1}^{N_2}(1+\rho\lambda_i)^{N_1} \prod_{j=1}^{N_1}\left(\xi_j^{(N_2-N_1)} e^{-\xi_j}e^{-(\rho \xi_j \lambda_{N_2+1-j})}\right) \prod_{j=1}^{N_1}\prod_{i=1}^{(N_2-j)}\left(\frac{1-e^{-(\rho \xi_j \lambda_i)}}{\rho \xi_j\lambda_i}\right).
\end{equation*}
\end{lemma}
Also, since $|\mathbf{V}_1(\bar{\xi})|=\prod_{i<j}^{N_2}(\xi_i-\xi_j)$ ~\cite{HJ}, for asymptotic high $\rho$, using the ordering among the $\xi_i$'s and equation \eqref{eq_asymptotic_def}, we have $|\mathbf{V}_1(\bar{\xi})|\dot{=} \prod_{i=1}^{N_2}\xi_i^{(N_2-i)}$. Putting these simplified expressions for $\mathbf{V}_1(\bar{\xi})$ and $J_1$ in equation (\ref{thm1_pf_eqn3}) we get
\begin{equation}
\label{thm1_pf_eqn5}
\mathbf{f_2}(\bar{\xi}|\bar{\gamma})\dot{=}C_2  \prod_{i=1}^{N_2}(1+\rho\lambda_i)^{N_1} \prod_{j=1}^{N_1}\left(e^{-\xi_j}e^{-(\rho \xi_j \lambda_{N_2+1-j})}\xi_j^{(N_2+N_1-2j)}\right) \prod_{j=1}^{N_1}\prod_{i=1}^{(N_2-j)}\left(\frac{1-e^{-(\rho \xi_j \lambda_i)}}{\rho \xi_j\lambda_i}\right).
\end{equation}
Recall that to use Lemma~\ref{lem_gas} we had to assume $N_1<N_2$; in what follows, we consider the case when $N_1\geq N_2$. Denoting $\tilde{H_1}\triangleq \Sigma^{\frac{1}{2}}\hat{H}_1$, $V_1=\hat{H}_1^\dagger\Sigma \hat{H}_1$ can be alternatively written as $V_1=\tilde{H_1}^{\dagger}\tilde{H_1}$. Also, the eigenvalues of $V_1$ and $\tilde{H_1}\tilde{H_1}^{\dagger}$ are the same for each realization of $\tilde{H}_1$. Since we are only interested in the eigenvalues of $V_1$, it is sufficient to compute the joint pdf of the eigenvalues of $W_1=\tilde{H_1}\tilde{H_1}^{\dagger}$. Now, $\tilde{H_1}\in \mathbb{C}^{N_2\times N_1}$ can be thought as a channel matrix, semi-correlated at the receiver, of an $N_2$-receive and $N_1$-transmit antenna MIMO channel. The conditional eigen-value distribution of $W_1$ given the eigenvalues- $\bar{\sigma} \triangleq [\sigma_1, \sigma_2, \cdots \sigma_{N_2}]$- of the correlation matrix, was found in~\cite{CZZ} for $N_2\leq N_1$ and is given by the following
\begin{equation}
\label{thm1_pf_eqn1}
\mathbf{f_2}(\bar{\xi}|\bar{\sigma})=C |\Sigma|^{-N_1}|W_1|^{N_1-N_2} |\mathbf{V}_1(\bar{\xi})|^2 \frac{|D(\bar{\xi},\bar{\sigma})|}{|\mathbf{V}_1(\bar{\xi})| |\mathbf{V}_2(\bar{\sigma})|}
\end{equation}
where $\xi_1 > \xi_2 > \cdots > \xi_{N_2} > 0$ are the ordered non-zero (w.p.1) eigenvalues of both $\tilde{H_1}\tilde{H_1}^\dagger$ and $V_1$, $\sigma_1>\sigma_2\cdots >\sigma_{N_2}>0$ are the ordered eigenvalues of $\Sigma$, $C$ is a constant independent of all the eigenvalues and
\begin{IEEEeqnarray}{rl}
D(\bar{\xi},\bar{\sigma})\triangleq & \left[e^{-\frac{\xi_j}{\sigma_i}}\right]_{i,j=1}^{N_2,N_2}, \nonumber\\
\mathbf{V}_2(\bar{\sigma})=&\mathbf{V}_1(-[\sigma_1^{-1}, \sigma_2^{-1}, \cdots \sigma_{N_2}^{-1}]).
\end{IEEEeqnarray}
In what follows, we will simplify the expression given in equation (\ref{thm1_pf_eqn1}) for asymptotic high values of SNR ($\rho$), assuming that the eigenvalues $\{\xi_1, \cdots \xi_{N_2}\}$ and $\{\lambda_1, \cdots \lambda_{N_2}\}$ vary exponentially with $\rho$. For asymptotically high $\rho$, using the ordering among the $\xi_j$'s, $\lambda_i$s and equation \eqref{eq_asymptotic_def}, it can be easily shown that $|\mathbf{V}_1(\bar{\xi})|\dot{=} \prod_{i=1}^{N_2}\xi_i^{(N_2-i)}$ and $|\mathbf{V}_2(\bar{\sigma})|\dot{=} \prod_{i=1}^{N_2}(\rho \lambda_i)^{(N_2-i)}$. Finally, the term $D(\bar{\xi},\bar{\sigma})$ in equation (\ref{thm1_pf_eqn1}) can also be simplified using the following Lemma.
\begin{lemma}
\label{lemma_eig1_help2}
If the eigenvalues $\{\xi_1\geq  \cdots \geq \xi_{N_2}\}$ and $\{\lambda_1\geq  \cdots \geq \lambda_{N_2}\}$ vary exponentially with $\rho$, then for asymptotic high SNR we have the following identity
\begin{equation*}
D(\bar{\xi},\bar{\sigma})\triangleq |[e^{-\frac{\xi_j}{\sigma_i}}]_{i,j=1}^{N_2,N_2}|~\dot{=} \prod_{j=1}^{N_2}\left(e^{-\xi_j}e^{-(\rho \xi_j \lambda_{N_2+1-j})}\right) \prod_{j=1}^{N_2}\prod_{i=1}^{(N_2-j)}\left(1-e^{-(\rho \xi_j \lambda_i)}\right).
\end{equation*}
\end{lemma}
Putting these simplified expressions for $|\mathbf{V}_1(\bar{\xi})|$, $|\mathbf{V}_2(\bar{\sigma})|$ and $D(\bar{\xi},\bar{\sigma})$ in equation (\ref{thm1_pf_eqn1}) we get
\begin{equation}
\label{thm1_pf_eqn2}
\mathbf{f}(\bar{\xi}|\bar{\lambda})\dot{=} \left(\prod_{i=1}^{N_2}(1+\rho \lambda_i)^{N_1}\right)\left( \prod_{j=1}^{N_2}\left(e^{-\xi_j}e^{-(\rho \xi_j \lambda_{N_2+1-j})} \xi_j^{(N_1+N_2-2j)}\right)\right)  \prod_{j=1}^{N_2}\prod_{i=1}^{(N_2-j)}\left(\frac{1-e^{-(\rho \xi_j \lambda_i)}}{\rho \lambda_i \xi_j}\right).
\end{equation}

Comparing equations (\ref{thm1_pf_eqn5}) and (\ref{thm1_pf_eqn2}), the expression for the distribution for $N_2\leq N_3$, can be written as
\begin{equation}
\label{thm1_pf_eqn6}
\mathbf{f_2}(\bar{\xi}|\bar{\gamma})\dot{=} \prod_{i=1}^{N_2}(1+\rho\lambda_i)^{N_1} \prod_{j=1}^{q}\left(e^{-\xi_j}e^{-(\rho \xi_j \lambda_{N_2+1-j})}\xi_j^{(N_2+N_1-2j)}\right) \prod_{j=1}^{q}\prod_{i=1}^{(N_2-j)}\left(\frac{1-e^{-(\rho \xi_j \lambda_i)}}{\rho \xi_j\lambda_i}\right).
\end{equation}

Recall that, we assumed that $N_2\leq N_3$ in the foregoing analysis. This assumption was important since it implies that all the eigenvalues of $V_2$ are non-zero and distinct. However, if $N_2>N_3$, then exactly $(N_2-N_3)$ eigenvalues of $V_2$ are zero and we need to consider this case separately. Because if two or more $\lambda_i$s are zero then both the numerator and denominator of the terms $J_1$ and $J_2=\frac{D(\bar{\xi},\bar{\sigma})}{\mathbf{V}_2(\bar{\sigma})}$ of equations \eqref{thm1_pf_eqn3} and \eqref{thm1_pf_eqn1}, respectively become zero. On the other hand, if only one of the $\lambda_i$s is zero, we can no longer assume that it varies with SNR. In either case the foregoing analysis cannot be pursued when $N_2>N_3$. However, both these problems can be overcome by doing the analysis in the limit when these $(N_2-N_3)$ eigenvalues tend to zero. It can be easily shown (e.g., see Lemma~6 in \cite{SMM}) that both $J_1$ and $J_2$ are well-defined in the limit when $(N_2-N_3)$ smallest eigenvalues of $V_2$ tend to zero, i.e.,
\begin{IEEEeqnarray}{l}
\lim_{\lambda_{(N_3+1)},\lambda_{(N_3+1)},\cdots ,\lambda_{N_2}\to 0}J_i=\lim_{\sigma_{1},\sigma_{2},\cdots ,\sigma_{(N_2-N_3)}\to 1}J_i
\end{IEEEeqnarray}
is well defined, for both $i=1,2$. Using this fact and following a similar approach as before we get
\begin{eqnarray}
\mathbf{f_2}(\bar{\xi}|\bar{\lambda})\dot{=}\lim_{(\lambda_{(N_3+1)}, \cdots \lambda_{(N_2)})\to 0}   \prod_{i=1}^{N_2}(1+\rho\lambda_i)^{N_1} \prod_{j=1}^{q}\left(e^{-\xi_j}e^{-(\rho \xi_j \lambda_{N_2+1-j})}\xi_j^{(N_2+N_1-2j)}\right) \prod_{j=1}^{q}\prod_{i=1}^{(N_2-j)}\left(\frac{1-e^{-(\rho \xi_j \lambda_i)}}{\rho \xi_j\lambda_i}\right) \nonumber\\
\label{thm1_pf_eqn7}
\dot{=}  \prod_{i=1}^{N_3}(1+\rho\lambda_i)^{N_1} \prod_{j=1}^{q}\left(e^{-\xi_j}e^{-(\rho \xi_j \tilde{\lambda}_{N_2+1-j})}\xi_j^{(N_2+N_1-2j)}\right) \prod_{j=1}^{q}\prod_{i=1}^{(N_2-j)\land N_3}\left(\frac{1-e^{-(\rho \xi_j \lambda_i)}}{\rho \xi_j\lambda_i}\right),
\end{eqnarray}
where
\begin{equation*}
\tilde{\lambda}_{N_2+1-j}=\left\{
\begin{array}{ll}
\lambda_{(N_2+1-j)}, &~\textrm{if}~ (N_2+1-j)\leq N_3;\\
0, &~\textrm{Otherwise}.
\end{array}\right.
\end{equation*}
Combining equations (\ref{thm1_pf_eqn6}) and (\ref{thm1_pf_eqn7}) Theorem~\ref{thm_eigenvalue_distribution_1} is proved.

\section{Proof of Theorem~\ref{thm_eigenvalue_distribution_2}}
\label{pf_thm_eigenvalue_distribution_2}

The ordering among the $\alpha_i$'s and $\beta_j$'s in $\mathcal{A}$ follows from the ordering of the eigenvalues $\lambda_i$'s and $\xi_j$'s, respectively. Besides the ordering, if $(\bar{\alpha},\bar{\beta})\notin \mathcal{A}$ then one or more of the following are true
\begin{IEEEeqnarray*}{rl}
e^{-\rho \lambda_i \xi_j}=&e^{-\rho^{(1-\alpha_i-\beta_j)}}\dot{=}~ 0,~\textrm{if}~(\alpha_i+\beta_j)<1;\\
e^{-\lambda_1}=&e^{-\rho^{-\alpha_1}}\dot{=}~ 0~\textrm{if}~\alpha_1<0;\\
e^{-\xi_1}=&e^{-\rho^{-\beta_1}}\dot{=} ~0~\textrm{if}~\beta_1<0;
\end{IEEEeqnarray*}
Since each of the terms on the left hand side of the above equations is a multiplying factor to the asymptotic expression of the joint pdf \eqref{joint_distribution}, it becomes zero if $(\bar{\alpha},\bar{\beta})\notin \mathcal{A}$. On the other hand if $(\bar{\alpha},\bar{\beta})\in \mathcal{A}$, then
\begin{IEEEeqnarray*}{rl}
(1+\rho\lambda_i)\dot{=}&\rho^{-(1-\alpha_i)^+},\\
e^{-\rho^{-\alpha_i}}\dot{=}&1,~\forall i,\\
e^{-\rho^{-\beta_j}}\dot{=}&1,~\forall j,\\ e^{-\rho^{(1-\alpha_i-\beta_j)}}\dot{=}&1,~\forall (i+j)=(N_2+1).
\end{IEEEeqnarray*}
Using this along with the Jacobians of the transformations $\lambda_i=\rho^{-\alpha_i}, ~1\leq i\leq p$ and $\xi_j=\rho^{-\beta_j}, ~1\leq j\leq q$ which are given by $J(\alpha_i)=\log(\alpha_i)\rho^{-\alpha_i}$ and $J(\beta_j)=\log(\beta_j)\rho^{-\beta_j}$, respectively, into equation (\ref{joint_distribution}) we get
\begin{eqnarray}
\label{exponent_distribution_1}
\mathbf{f}(\bar{\alpha},\bar{\beta})\dot{=}\rho^{-\Big(\sum_{i=1}^{p}\left((N_3+N_2-2i+1)\alpha_i-N_1(1-\alpha_i)^+\right)
+\sum_{j=1}^{q}(N_1+N_2-2j+1)\beta_j \Big)} \prod_{j=1}^{q}\prod_{i=1}^{(N_2-j)\land N_3}\left(\frac{1-{e}^{-\rho^{1-\alpha_i-\beta_j} }}{\rho^{1-\alpha_i-\beta_j}}\right).
\end{eqnarray}
Also, at high SNR limit for any $i,j$
\begin{equation*}
\left(\frac{1-{e}^{-\rho^{1-\alpha_i-\beta_j} }}{\rho^{1-\alpha_i-\beta_j}}\right)\dot{=}
\left\{\begin{array}{c}
1, ~~~~~~~~~~\textrm{if} ~(\alpha_i+\beta_j)>1;\\
\rho^{-(1-\alpha_i-\beta_j)}, \textrm{if} ~(\alpha_i+\beta_j)\leq 1.
\end{array}\right.~\left[\because~ \lim_{z\to 0}\frac{1-e^{-z}}{z}=1\right].
\end{equation*}
Now, using the fact that the product of several converging sequences converges to the product of their individual limiting values we get
\begin{equation}
\prod_{j=1}^{q}\prod_{i=1}^{(N_2-j)\land N_3}\left(\frac{1-{e}^{-\rho^{1-\alpha_i-\beta_j} }}{\rho^{1-\alpha_i-\beta_j}}\right)\dot{=}\prod_{j=1}^{q}\prod_{i=1}^{(N_2-j)\land N_3}\rho^{-(1-\alpha_i-\beta_j)^+}.
\end{equation}
Finally, putting this in equation \eqref{exponent_distribution_1} we get Theorem~\ref{thm_eigenvalue_distribution_2}.

\section{Proof of theorem~\ref{thm_optimization_problem}}
\label{pf_thm_optimization_problem}

As discussed before, here we shall compute the SNR exponent of $\Pr\left(\mathcal{O}|f<1\right)$. The proof of this part is rather long and hence divided into three steps. We start with an outline of the different parts. In the first part a straightforward analysis of the outage probability following a similar method as in \cite{tse1} yields $\hat{d}(r)$ as the minimum value of the negative SNR exponent of the corresponding pdf, where the minimization is over the intersection of the outage set and the support set of the pdf. In the next step, this problem is then transformed into an equivalent one having smaller number of variables which is solved in the final and third step.

{\bf Step 1:}~ Using Laplace's method of integration as in \cite{tse1} we get from equation \eqref{outage_integration}
\begin{IEEEeqnarray}{rl}
\hat{d}(r)=\min_{\left\{\left(\bar{\alpha},\bar{\beta},\bar{\gamma}\right)\in \mathcal{O}\cap \mathcal{S}\right\}}\left(E\left(\bar{\alpha},\bar{\beta}\right) +\sum_{i=1}^{t}(k+m-2i+1)\gamma_i\right)=\min_{\left\{\left(\bar{\alpha},\bar{\beta},\bar{\gamma}\right)\in \mathcal{O}\cap \mathcal{S}\right\}}F\left(\bar{\alpha},\bar{\beta},\bar{\gamma}\right),
\end{IEEEeqnarray}
where $\mathcal{S}=\mathcal{A}\cap \mathcal{D}$ represents the support set of the pdf of $\left(\bar{\alpha},\bar{\beta},\bar{\gamma}\right)$.

Suppose at a given $r$, the objective function attains the minimum value for an $\bar{\alpha}\in \mathcal{A}\cap \mathcal{D}$ where $\alpha_i>1$ for one or more $i$'s. Let $\widetilde{\bar{\alpha}}=\min \{[1, 1, \cdots ,1],\bar{\alpha}\}$, where the minimization is componentwise. Clearly, $\widetilde{\bar{\alpha}}\in \mathcal{A}\cap \mathcal{D}$ but at this point $E$ has a strictly smaller value. This proves that in the optimal solution, $\alpha_i\in [0, 1]$ for all $i$. The same is true for $\bar{\beta}$ and $\bar{\gamma}$. Thus $\hat{d}(r)$ is given as

\begin{IEEEeqnarray}{rl}
\label{general_opt}
\hat{d}(r)= \min_{\left\{\left(\bar{\alpha},\bar{\beta},\bar{\gamma}\right)\in \hat{\mathcal{O}}\right\}}\sum_{i=1}^{p}(n+m+k-2i+1)\alpha_i -kp &+ \sum_{j=1}^{q}(n+k-2j+1)\beta_j+\sum_{j,i=1}^{q,(n-j)\land m}(1-\alpha_i-\beta_j)^+ \nonumber\\
\label{eq_opt_1}
& +\sum_{i=1}^{(k\land m)}(k+m-2i+1)\gamma_i,
\end{IEEEeqnarray}
where
\begin{IEEEeqnarray}{rl}
\hat{\mathcal{O}}=\Bigg\{ \left(\bar{\alpha},\bar{\beta},\bar{\gamma}\right):
\label{mi_constraint}
\sum_{i=1}^{p}(1-\alpha_i)+(1-f)\sum_{j=1}^{q}(1-\beta_j)\leq & r,\\
\label{ddf_constraint}
0\leq f = \frac{r}{\sum_{i=1}^{t}(1-\gamma_i)}< & 1,\\
\label{redun1}
\left(\alpha_{i}+\beta_j\right) \geq &1, \forall (i,j):(i+j)\geq (n+1),\\
\label{redun2}
0\leq \alpha_1 \leq \cdots \leq \alpha_p \leq &1,\\
0\leq \beta_1 \leq \cdots \leq \beta_{q} \leq &1, \\
\label{redun4}
0\leq \gamma_1 \leq \cdots \leq \gamma_{t} \leq &1\Bigg\},
\end{IEEEeqnarray}
where equation \eqref{ddf_constraint} follows from equation \eqref{ratio} and the fact that $f<1$. Note that the number of optimization variables increases linearly with the number of antennas at all the nodes. To overcome this problem, in what follows, we transform the previous optimization problem into an equivalent one having a fixed number of variables which is independent of the number of antennas at the nodes.
\\
{\bf Step~2:}~The objective function in \eqref{eq_opt_1} decreases strictly monotonically as $\alpha_i$ is decreased for any $i$ and the rate of decrease with $\alpha_i$ is smaller for a larger value of $i$. The same is true for $\bar{\beta}$ and $\bar{\gamma}$. Thus, following a similar method as in $\cite{tse1}$, it can be shown that if $\sum_{i=1}^{p}(1-\alpha_i)=a$, $\sum_{j=1}^{q}(1-\beta_j)=b$, $\sum_{l=1}^{t}(1-\gamma_l)=y$ and $\left(\bar{\alpha},\bar{\beta},\bar{\gamma}\right)$ satisfy equations \eqref{redun2}-\eqref{redun4}, then the optimal choice of $\left(\bar{\alpha},\bar{\beta},\bar{\gamma}\right)$ that minimizes $F(.)$ is given by $(\phi_{\alpha}(a),\phi_{\beta}(b),\phi_{\gamma}(y))$, where
\begin{IEEEeqnarray}{l}
\label{eq_phi_alpha}
\phi_{\alpha}(a)=[\hat{\alpha}_1, \hat{\alpha}_2, \cdots, \hat{\alpha}_p ]^T~ :~\hat{\alpha}_i=\left(1-\left(a-i+1\right)^+\right)^+, ~1\leq i\leq p.
\end{IEEEeqnarray}
\begin{IEEEeqnarray}{l}
\label{eq_phi_beta}
\phi_{\beta}(b)=[\hat{\beta}_1, \hat{\beta}_2, \cdots, \hat{\beta}_q ]^T~ :~\hat{\beta}_j=\left(1-\left(b-j+1\right)^+\right)^+, ~1\leq j\leq q;\\
\label{eq_phi_gamma}
\phi_{\gamma}(y)=[\hat{\gamma}_1, \hat{\gamma}_2, \cdots, \hat{\gamma}_t ]^T~ :~\hat{\gamma}_l=\left(1-\left(y-l+1\right)^+\right)^+, ~1\leq l\leq t.
\end{IEEEeqnarray}

Denoting by $\mathcal{T}(a,b,y)$ the following set
\begin{IEEEeqnarray*}{rl}
\Big\{(\bar{\alpha}, \bar{\beta},\bar{\gamma}):& \sum_{i=1}^{p}(1-\alpha_i)=a, \sum_{j=1}^{q}(1-\beta_j)=b, \sum_{l=1}^{q}(1-\delta_l)=y,~\left(\bar{\alpha},\bar{\beta},\bar{\gamma}\right) ~\textrm{satisfy equations \eqref{redun2}-\eqref{redun4}}\Big\},
\end{IEEEeqnarray*}
from the foregoing argument we have
\begin{equation}
\label{eq_pf_opt_problem_first_min}
\min_{\{\mathcal{T}(a,b,y)\}} F\left((\bar{\alpha}, \bar{\beta},\bar{\gamma})\right)=F\left(\phi_{\alpha}(a),\phi_{\beta}(b),\phi_{\gamma}(y)\right).
\end{equation}
Let us now define the following set of new variables
\begin{IEEEeqnarray}{l}
\mathcal{O}_1= \left\{(a,b,y): a+b\left(1-\frac{r}{y}\right)\leq r,~ (a+b)\leq n,~ 0\leq b\leq q,~r< y\leq t\right\}.
\end{IEEEeqnarray}
It is clear from the definition of $\mathcal{T}(a,b,y)$ that,
\begin{equation}
\hat{\mathcal{O}}_1 \triangleq \bigcup_{\{(a,b,y)\in \mathcal{O}_1\}}\mathcal{T}\left(a,b,y\right)\supset \hat{\mathcal{O}}.
\end{equation}
Since the minimum of a function over a set is not larger than the minimum of that function over a subset of it, the above relation along with equation~\eqref{eq_pf_opt_problem_first_min} imply
\begin{equation}
\label{eq_equivalent_opt_lb}
\min_{\{(a,b,y)\in \mathcal{O}_1\}} F\left(\phi_{\alpha}(a),\phi_{\beta}(b),\phi_{\gamma}(y)\right)= \min_{\{(\bar{\alpha}, \bar{\beta},\bar{\gamma})\in \hat{\mathcal{O}}_1\}} F\left(\bar{\alpha}, \bar{\beta},\bar{\gamma}\right) \leq \min_{\{(\bar{\alpha}, \bar{\beta},\bar{\gamma})\in \hat{\mathcal{O}}\}} F\left(\bar{\alpha}, \bar{\beta},\bar{\gamma}\right).
\end{equation}
Before proceeding further we take note of a few properties of the newly defined variables $a, b$ and $y$. From the definition of $\phi_i$'s it is clear that if $(a,b,y)\in \mathcal{O}_1$, then $\left(\phi_{\alpha}(a),\phi_{\beta}(b),\phi_{\gamma}(y)\right)$ satisfy equations \eqref{mi_constraint}, \eqref{ddf_constraint} and \eqref{redun2}-\eqref{redun4}. Suppose for some $(i+j)=(n+1)$, $(\hat{\alpha}_i+\hat{\beta_j})<1$, then it can be shown that $\sum_{i=1}^{p}(1-\hat{\alpha}_i)+\sum_{j=1}^{q}(1-\hat{\beta}_j)> n$, which is impossible. Thus $(\hat{\alpha}_i+\hat{\beta_j})\geq 1$ for all $(i+j)\geq 1$, which in turn imply that the $\left(\phi_{\alpha}(a),\phi_{\beta}(b),\phi_{\gamma}(y)\right)$ tuple also satisfies equation \eqref{redun1}. That is
\begin{IEEEeqnarray*}{l}
(a,b,y)\in \mathcal{O}_1 ~\Rightarrow~ (\phi_{\alpha}(a),\phi_{\beta}(b),\phi_{{\gamma}}(y))\in \hat{\mathcal{O}},\\
\label{eq_equivalent_opt_ub}
\Rightarrow~\min_{\{(\bar{\alpha}, \bar{\beta},\bar{\gamma})\in \hat{\mathcal{O}}\}} F\left(\bar{\alpha}, \bar{\beta},\bar{\gamma}\right) \leq \min_{\{(a,b,y)\in \mathcal{O}_1\}} F\left(\phi_{{\alpha}}(a),\phi_{{\beta}}(b),\phi_{{\gamma}}(y)\right).
\end{IEEEeqnarray*}
Combining this with equation \eqref{eq_equivalent_opt_lb}, we get
\begin{equation}
\label{eq_equivalent_opt_1}
\min_{\{(\bar{\alpha}, \bar{\beta},\bar{\gamma})\in \hat{\mathcal{O}}\}} F\left(\bar{\alpha}, \bar{\beta},\bar{\gamma}\right) ~=~ \min_{\{(a,b,y)\in \mathcal{O}_1\}} F\left(\phi_{\alpha}(a),\phi_{\beta}(b),\phi_{\gamma}(y)\right).
\end{equation}
Therefore, we have an equivalent optimization problem to that presented in equation \eqref{eq_opt_1}, but with less number of variables, i.e., $\hat{d}(r)$ can be equivalently written as
\begin{equation}
\hat{d}(r)=\min_{\left\{(a,b,y)\in \mathcal{O}_1 \right\}}~ F\left(\phi_{{\alpha}}(a),\phi_{{\beta}}(b),\phi_{{\gamma}}(y)\right).
\end{equation}

When $(a+b)=n$, the objective function has a property which we state now and will be helpful to solve the minimization problem in the next section.
\begin{cl}
\label{cl_a_plus_b_equal_1}
$F\left(\phi_{{\alpha}}(a),\phi_{{\beta}}(n-a),\phi_{{\gamma}}(y)\right)$ is monotonically decreasing with $a$.
\end{cl}
\begin{proof}
It can be shown using equations \eqref{eq_phi_alpha}-\eqref{eq_phi_gamma} that when $(a+b)=n$ we have
\begin{equation*}
\label{eq_iplusj_equalto1}
    (\hat{\alpha}_i+\hat{\beta}_j)=1, ~\forall (i+j)=(n+1)~\textrm{and}~
    (\hat{\alpha}_i+\hat{\beta}_j)\leq 1, ~\forall (i+j)\leq n.
\end{equation*}
Using these relations in the expression for $F\left(\bar{\hat{\alpha}},\bar{\hat{\beta}},\bar{\hat{\gamma}}\right)$, after some algebra we get
\begin{IEEEeqnarray*}{rl}
F\left(\bar{\hat{\alpha}},\bar{\hat{\beta}},\bar{\hat{\gamma}}\right)=& \sum_{i=1}^{p}(m+n+1-2i)\hat{\alpha}_i+\sum_{l=1}^{t}(m+k+1-2l)\hat{\gamma}_l,\\
=&\sum_{i=1}^{p}(m+n+1-2i)\left(1-\left(a-i+1\right)^+\right)^++\sum_{l=1}^{t}(m+k+1-2l)\hat{\gamma}_l.
\end{IEEEeqnarray*}
The above expression is clearly independent of $b$ and monotonically decreasing with $a$.
\end{proof}

{\bf Step 3:} In this final step, we determine the minimum of $F(.)$ on $\mathcal{O}_1$ and establish the theorem. Depending on the value of $y$ the set of feasible $(a,b)$ pairs takes on different shapes as shown in the following figures. For example, when $\frac{qr}{(n-r)}<y\leq \frac{qr}{(q-r)}$ the feasible set of $(a,b)$ pairs is the trapezoid ABDE shown in Figure~\ref{y_range-b}.

\begin{figure}[htp]
\subfigure[$\mathcal{R}_1=\left\{r<y\leq \frac{qr}{(n-r)}\right\}.$]{\label{y_range-a}\setlength{\unitlength}{1mm}\begin{picture}(70,40)
\thicklines
\put(0,0){\vector(1,0){50}}
\put(45,-4){$b$}
\put(0,0){\vector(0,1){40}}
\put(-4,35){$a$}
\put(0,22){\line(2,-1){40}}
\put(25,10){$a+b(1-\frac{r}{y})=r$}
\put(20,0){\line(0,1){20}}
\put(11,3){$b=q$}
\put(0,28){\line(1,-1){28}}
\put(6,23){$(a+b)=n$}
\put(-1,-3){E}
\put(19,-3){A}
\put(-3,21){D}
\put(20,8){B}
\put(12,16){C}
\end{picture}}
\subfigure[$\mathcal{R}_2=\left\{\frac{qr}{(n-r)}<y\leq \frac{qr}{(q-r)}\right\}.$]{\label{y_range-b}\setlength{\unitlength}{1mm}\begin{picture}(95,40)
\thicklines
\put(0,0){\vector(1,0){60}}
\put(55,-4){$b$}
\put(0,0){\vector(0,1){40}}
\put(-4,35){$a$}
\put(0,22){\line(4,-3){29}}
\put(25,4){$a+b(1-\frac{r}{y})=r$}
\put(20,0){\line(0,1){20}}
\put(21,15){$b=q$}
\put(0,31){\line(1,-1){23}}
\put(6,27){$(a+b)=n$}
\put(-1,-3){E}
\put(19,-3){A}
\put(-3,21){D}
\put(20,7){B}
\end{picture}}
\subfigure[$\mathcal{R}_3=\left\{\frac{qr}{(q-r)}<y\leq t\right\}.$]{\label{y_range-c}\setlength{\unitlength}{1mm}\begin{picture}(70,40)
\thicklines
\put(0,0){\vector(1,0){50}}
\put(45,-4){$b$}
\put(0,0){\vector(0,1){40}}
\put(-4,35){$a$}
\put(0,22){\line(3,-4){16.5}}
\put(14,5){$a+b(1-\frac{r}{y})=r$}
\multiput(20,0)(0,2){10}{\line(0,1){.5}}
\put(0,31){\line(1,-1){23}}
\put(6,27){$(a+b)=n$}
\put(-1,-3){E}
\put(19,-3){A}
\put(-3,21){D}
\put(16,-3){F}
\end{picture}}
{\center \caption{Sets of feasible $(a,b)$ tuples for different range of $y$.}}
\label{y_range}
\end{figure}
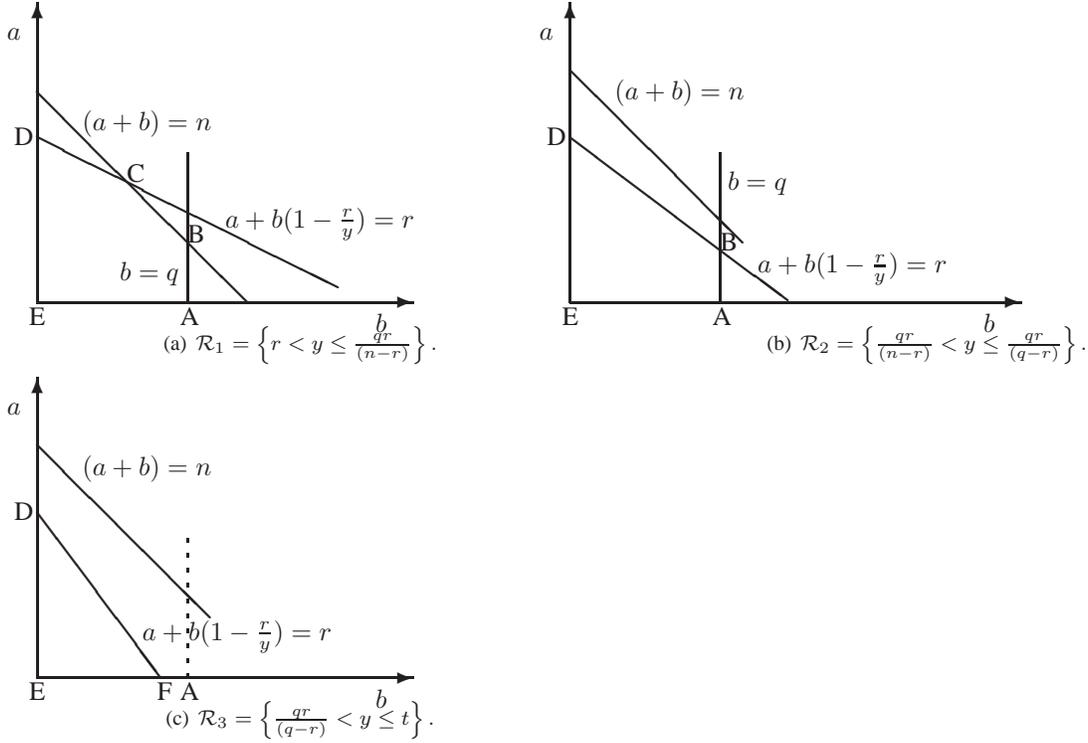

For any given value of $y$ the following observations will help us solve the problem:
\begin{itemize}
\item  The optimal $(a,b)$ pair always lies on the boundary, because the objective function is monotonically decreasing with both $a$ and $b$.
\item By the same argument the optimal point on the line segment AB is B.
\item The optimal point on the line segment BC is C. Because, by Claim~\ref{cl_a_plus_b_equal_1} when $(a+b)=n$, the objective function is independent of $b$ and monotonically decreasing with $a$.
\end{itemize}
Thus for a given $y$, the objective function has to be minimized on the line segment CD, BD or DF. In what follows, we treat each of these cases individually:
\begin{enumerate}
\item When $\mathcal{R}_1=\left(r, \frac{qr}{(n-r)}\right]$, the optimal $(a,b)$ lies on DC and the diversity order is given by
\begin{equation*}
d^*(r)=\min_{\left\{y\in \mathcal{R}_1,~0\leq b\leq \frac{y(n-r)}{r}\right\}} F\left(\phi_{{\alpha}}(a),\phi_{{\beta}}(b),\phi_{{\gamma}}(y)\right).
\end{equation*}

\item When $\mathcal{R}_2=\left(\frac{qr}{(n-r)}, \frac{qr}{(q-r)}\right]$, the optimal $(a,b)$ lies on BD and the diversity order is given by
\begin{equation*}
d^*(r)=\min_{\left\{y\in \mathcal{R}_2,~0\leq b\leq q\right\}} F\left(\phi_{{\alpha}}(a),\phi_{{\beta}}(b),\phi_{{\gamma}}(y)\right).
\end{equation*}

\item When $\mathcal{R}_3=\left(\frac{qr}{(q-r)}, t\right]$, the optimal $(a,b)$ lies on DF and the diversity order is given by
\begin{equation*}
d^*(r)=\min_{\left\{y\in \mathcal{R}_3,~0\leq b\leq \frac{ry}{(y-r)}\right\}} F\left(\phi_{{\alpha}}(a),\phi_{{\beta}}(b),\phi_{{\gamma}}(y)\right).
\end{equation*}
\end{enumerate}
where $a+b\left(1-\frac{r}{y}\right)=r$ in all of the above cases. Finally, combining the different cases we get
\begin{equation*}
\hat{d}(r)= \min_{1\leq i\leq 3}~\min_{\left\{y\in \mathcal{R}_i,~b\in \mathcal{B}_i(y)\right\}} F\left(\phi_{{\alpha}}\left(r-b\left(1-\frac{r}{y}\right)\right), \phi_{{\beta}}(b), \phi_{{\gamma}}\left(y\right) \right),
\end{equation*}
where $\mathcal{B}_1(y)=\left[0,\frac{y(n-r)}{r}\right];~
\mathcal{B}_2(y)=\left[0,q\right];~
\mathcal{B}_3(y)=\left[0, \frac{ry}{(y-r)}\right]$.

\section{Proof of Theorem~\ref{thm_analytical1}: The DMT on a ($n,1,n$) relay channel}
\label{pf_thm_analytical1}
We consider the case where $n=m\geq 2$ and $k=1$. Combining it with Example~\ref{ex_dmt_111_ch}, the proof of the theorem will be complete. From Theorem~\ref{thm_optimization_problem} we know that when $r\geq t=1$, $d_1^*(r)=d_{n,n}(r)$. So, let us consider the case when $r\leq 1$. Since the optimal solution always lie on the line $a+b\left(1-\frac{r}{y}\right)=r$, $r\leq 1$ implies that $a\leq 1$. Using this in the definitions of $\phi_i$s, we get
\begin{IEEEeqnarray*}{rl}
G(a,b,y)=& F\left(\phi_{{\alpha}}(a),\phi_{{\beta}}(b),\phi_{{\gamma}}(y)\right),\\
=& n(n-1)+2n(1-a)+n(1-b)+n(1-y)-n+(a+b-1)^+.
\end{IEEEeqnarray*}
We know from Theorem~\ref{thm_optimization_problem} that the above objective function has to be minimized over three different sets of $(a,b)$ pairs. For $r\leq 1$, $\mathcal{R}_1=\Phi$, so we consider $y\in \mathcal{R}_2=\left(\frac{r}{n-r},\frac{r}{1-r}\right]$. We know from Figure~\ref{y_range-b} that the optimal solution lie on the line segment BD and since the objective function is also linear in $a$ and $b$ the optimal solution is one of the extreme points depending on the slope of the line
\begin{equation*}
a+b\left(1-\frac{r}{y}\right)=r.
\end{equation*}
Since
\begin{equation*}
\left(1-\frac{r}{y}\right)\leq \frac{1}{2}~\forall~y~\in \mathcal{R}_2,
\end{equation*}
the objective function attains its minimum at B, where $b^*=1$ and $a^*=r-\left(1-\frac{r}{y}\right)$. Putting this into the objective function we have
\begin{equation}
\label{eq_dmt_n1n_temp1}
G(a^*,b^*,y)=(n-1)^2+(2n-1)(1-r)+(2n-1)\left(1-\frac{r}{y}\right)+n(1-y).
\end{equation}
The above function is convex in $y$ and the infimum is attained at $y=r$ which is given as
\begin{equation}
\label{eq_dmt_n1n_R2}
d\left(\mathcal{R}_2\right)=(n-1)^2+(3n-1)(1-r),~0\leq r\leq 1.
\end{equation}

Next we consider the case when $y\in \mathcal{R}_3=\left(\frac{r}{1-r},1\right]$. This set is non-empty only for $r\leq \frac{1}{2}$ and the optimal point lie on the line segment DF in Figure~\ref{y_range-c}. Dividing the set $\mathcal{R}_3$ further into two subsets $\mathcal{R}_{31}=\left(\frac{r}{1-r},2r\right]$ and $\mathcal{R}_{32}=(2r,1]$ we see that $\left(1-\frac{r}{y}\right)\leq \frac{1}{2}$ when $y\in \mathcal{R}_{31}$ and $\left(1-\frac{r}{y}\right)\geq \frac{1}{2}$ when $y\in \mathcal{R}_{32}$. The objective function attains minimum value at point F when $y\in \mathcal{R}_{31}$ and at point D when $y\in \mathcal{R}_{32}$ and given by
\begin{equation*}
G(a^*,b^*,y)=\left\{\begin{array}{cc}
n^2+n\left(2-\frac{y^2}{y-r}\right),&\textrm{for}~y\in \mathcal{R}_{31};~\left[\because~(a^*,b^*)=\left(0,\frac{yr}{y-r}\right)\right]\\
n(n-1)+2n(1-r)+(1-y),&\textrm{for}~\textrm{for}~y\in \mathcal{R}_{32},~\left[\because~(a^*,b^*)=\left(r,0\right)\right].
\end{array}\right.
\end{equation*}
Now, optimizing this function in the corresponding sets of $y$ we get
\begin{equation*}
G(a^*,b^*,y^*)=\left\{\begin{array}{cc}
n^2+2n(1-2r),&~\left[\because ~y^*=2r,~\textrm{for}~y\in \mathcal{R}_{31}\right]\\
n(n-1)+2n(1-r),&~\left[\because ~y^*=1,~\textrm{for}~y\in \mathcal{R}_{32}\right].
\end{array}\right.
\end{equation*}
which in turn imply
\begin{equation}
\label{eq_dmt_n1n_R3}
d\left(\mathcal{R}_3\right)=n(n-1)+2n(1-r),~\textrm{for}~0\leq r\leq \frac{1}{2}.
\end{equation}
Now, combining equations \eqref{eq_dmt_n1n_R2} and \eqref{eq_dmt_n1n_R3} we get
\begin{equation}
\label{eq_dmt_n1n_R23}
\hat{d}(r)=\min \left\{d\left(\mathcal{R}_2\right), d\left(\mathcal{R}_3\right)\right\}=(n-1)^2+(3n-1)(1-r),~\textrm{for}~0\leq r\leq 1.
\end{equation}
Finally, putting this in Theorem~\ref{thm_optimization_problem} and combining the result with Example~\ref{ex_dmt_111_ch} we get
\begin{equation*}
\label{eq_dmt_n1n_rlessthan1}
d^*_1(r)=\left\{\begin{array}{cc}
\frac{(1-r)}{\max\{\frac{1}{2},r\}},&~0\leq r\leq 1~\textrm{and}~ n=1;\\
(n-1)^2+(3n-1)(1-r),&~0\leq r\leq 1~\textrm{and}~ n\geq 2;\\
d_{n,n}(r),&~1\leq r\leq n~\textrm{and} ~n\geq 2.
\end{array}\right.
\end{equation*}

\section{Proof of Theorem~\ref{thm_analytical2}: The DMT on a $(1,k,1)$ relay channel}
\label{pf_thm_analytical2}
On a $(1,k,1)$ channel we have $a, b, y\leq 1$. Putting this in the definitions of $\phi_i$s from Theorem~\ref{thm_optimization_problem} we get
\begin{equation*}
G(a,b,y)=F\left(\phi_{{\alpha}}(a),\phi_{{\beta}}(b),\phi_{{\gamma}}(y)\right)=(k+1)(1-r)+(k+1)b\left(1-\frac{r}{y}\right)-kb+k(1-y).
\end{equation*}
Since
\begin{equation*}
\left(1-\frac{r}{y}\right)\leq \frac{1}{2}~\forall~y~\in \mathcal{R}_1,
\end{equation*}
the objective function attains its minimum at $b^*=\max\{\mathcal{B}_1(y)\}=y\left(\frac{1-r}{r}\right)$. Putting this into the objective function we have
\begin{equation}
\label{eq_dmt_1k1_temp1}
G(a^*,b^*,y)=b^*+k(1-y)=k+y\left(\frac{1-(k+1)r}{r}\right).
\end{equation}
Clearly, the largest and smallest feasible value of $y$ in $\mathcal{R}_1=\left(r,\frac{r}{1-r}\right]$ minimizes the above function when the coefficient of $y$ is negative and non-negative, respectively. That is, the optimal $y$ is given as
\begin{equation*}
y^*=\left\{\begin{array}{cc}
r,&\textrm{for}~0\leq r\leq \frac{1}{(1+k)};\\
\frac{r}{1-r},&\textrm{for}~\frac{1}{(1+k)}\leq r \leq \frac{1}{2};\\
1,&\textrm{for}~r\geq \frac{1}{2}.
\end{array}\right.
\end{equation*}
Putting this solution in equation \eqref{eq_dmt_1k1_temp1} we get
\begin{equation}
\label{eq_dmt_1k1_R1}
d\left(\mathcal{R}_1\right)=\left\{\begin{array}{cc}
(k+1)(1-r),&\textrm{for}~0\leq r\leq \frac{1}{2};\\
1+k\left(\frac{1-2r}{1-r}\right), & \textrm{for}~\frac{1}{(k+1)}\leq r\leq \frac{1}{2};\\
\left(\frac{1-r}{r}\right),&\textrm{for}~\frac{1}{2}\leq r\leq 1.
\end{array}\right.
\end{equation}
Since for $m=n=1$, $\mathcal{R}_2=\Phi$, next we consider the case when $y\in \mathcal{R}_3=\left(\frac{r}{1-r},1\right]$. This set is non-empty only for $r\leq \frac{1}{2}$ and the optimal point lies on the line segment DF in Figure~\ref{y_range-c}. Dividing the set $\mathcal{R}_3$ further into two subsets namely $\mathcal{R}_{31}=\left(\frac{r}{1-r},(k+1)r\right]$ and $\mathcal{R}_{32}=\Big((k+1)r,1\Big]$, we see that $\left(1-\frac{r}{y}\right)\leq \frac{1}{2}$ when $y\in \mathcal{R}_{31}$ and $\left(1-\frac{r}{y}\right)\geq \frac{1}{2}$ when $y\in \mathcal{R}_{32}$. The objective function attains minimum value at point F, where $b=\max \mathcal{B}_3(y)$ when $y\in \mathcal{R}_{31}$ and at point D, where $b=\min \mathcal{B}_3(y)$ when $y\in \mathcal{R}_{32}$ and given by
\begin{equation*}
G(a^*,b^*,y)=\left\{\begin{array}{cc}
1+k\left(2-\frac{y^2}{y-r}\right),&\textrm{for}~y\in \mathcal{R}_{31};~\left[b^*=\frac{yr}{y-r}\right]\\
(k+1)(1-r)+k(1-y),&\textrm{for}~\textrm{for}~y\in \mathcal{R}_{32}.~\left[b^*=0\right].
\end{array}\right.
\end{equation*}
Both these functions are minimized by the maximum value of $y$ in their corresponding range which are
\begin{equation*}
y^*=\left\{\begin{array}{cc}
(k+1)r,&\textrm{when}~y\in \mathcal{R}_{31} \cap \left\{r\leq \frac{1}{(1+k)}\right\};\\
1,&\textrm{when}~y\in \mathcal{R}_{31} \cap \left\{r\geq \frac{1}{(1+k)}\right\};\\
1,&\textrm{when}~~y\in \mathcal{R}_{32} \cap \left\{r\leq \frac{1}{(1+k)}\right\},
\end{array}\right.
\end{equation*}
which in turn gives us
\begin{equation}
\label{eq_dmt_1k1_R3}
d\left(\mathcal{R}_3\right)=\left\{\begin{array}{cc}
(k+1)(1-r),~\textrm{for}~0\leq r\leq \frac{1}{(k+1)};\\
1+k\left(\frac{1-2r}{1-r}\right),~\textrm{for}~\frac{1}{(k+1)}\leq r\leq \frac{1}{2}.
\end{array}\right.
\end{equation}
Combining equations \eqref{eq_dmt_1k1_R1} and \eqref{eq_dmt_1k1_R3} we get
\begin{IEEEeqnarray*}{rl}
\hat{d}(r)=&\min \left\{d\left(\mathcal{R}_1\right),d\left(\mathcal{R}_3\right)\right\},\\
=&\left\{\begin{array}{cc}
(k+1)(1-r),&\textrm{for}~0\leq r\leq \frac{1}{2};\\
1+k\left(\frac{1-2r}{1-r}\right), & \textrm{for}~\frac{1}{(k+1)}\leq r\leq \frac{1}{2};\\
\left(\frac{1-r}{r}\right),&\textrm{for}~\frac{1}{2}\leq r\leq 1.
\end{array}\right.
\end{IEEEeqnarray*}
Finally, by Theorem~\ref{thm_optimization_problem} we have $d_2^*(r)=\hat{d}(r)$.

\section{Proof of Theorem~\ref{analytical3}: The DMT of the $(2,k,2)$ relay channel}
\label{pf_thm_analytical3}
Instead of optimizing the objective function over all possible values of $(a,b)$, we restrict the minimization over a set on which $a+b=n$ and thus obtain an upper bound. From the proof of Theorem~\ref{thm_optimization_problem} we know that the optimal $(a,b)$ pair lies on the line
\begin{equation}
\label{constraint_thm_analytic3_pf_1st}
a+b\left(1-\frac{r}{y}\right)=r.
\end{equation}
Combining it with the assumption just made, it is clear that the optimal choice of $a$, for any feasible $y$, is given as
\begin{equation}
a^*=2-\left(\frac{(2-r)}{r}\right)y.
\end{equation}
Now, by Claim~\ref{cl_a_plus_b_equal_1} we know that when $a+b=n$, the objective function becomes
\begin{equation*}
\sum_{i=1}^{2}(5-2i)\hat{\alpha}_i+\sum_{l=1}^{2}(k+3-2l)\hat{\gamma}_l.
\end{equation*}
Computing the values of $\bar{\hat{\alpha}}$ and $\bar{\hat{\gamma}}$ from the definitions of $\phi_i$s and substituting in the above equation we get
\begin{IEEEeqnarray}{rl}
d_u(r)=&\min_{r\leq y\leq q} \sum_{i=1}^{2}(5-2i)\left((i-a^*)\land 1\right)^+ + \sum_{l=1}^{q}(k+3-2l)\left((l-y)\land 1\right)^+,\nonumber\\
\label{eq_dmt_222_func_y}
=&\min_{r\leq y\leq q} \sum_{i=1}^{2}(5-2i)\left(\left(i-2+\left(\frac{2-r}{r}\right)y\right)\land 1\right)^+ + \sum_{l=1}^{q}(k+3-2l)\left((l-y)\land 1\right)^+.
\end{IEEEeqnarray}
To optimize this function with respect to $y$, in the following we divide the values of $a^*$ and $y$ in four different regions as follows
\begin{equation}
\mathcal{R}_{u,v}=\left\{(u-1)\leq a^*\leq u, ~~~(v-1)\leq y\leq v\right\}, ~~u,v=1,2.
\end{equation}
In any of these regions the function is minimized either if $y$ attains its maximum value, denoted as $y_M$ in that region or the minimum value, denoted as $y_m$ in that region. However, in a particular region, the extreme values depend also on the value of $r$. So in the following, depending on $r$ each region will be further divided into several sub-regions and in each sub-region the two extreme values of $y$ will yield two values of the objective function as functions of $r$. We will have to take the minimum among all these functions to get the desired minimum of the objective function.

{\it Region $\mathcal{R}_{1,1}$:} In this region the set of feasible values of $y$ is given by the following
\begin{equation*}
\left\{\frac{r}{2-r}\leq y\leq \frac{2r}{2-r}\right\}\cap \Big\{0\leq y\leq 1\Big\}\cap \Big\{r\leq y\Big\}.
\end{equation*}
Note for $0\leq r\leq 1$, $y_m=r$, for $0\leq r\leq \frac{2}{3}$, $y_M=\frac{2r}{2-r}$, for $\frac{2}{3}\leq r\leq 1$ $y_M=1$ and for $r\geq 1$ the above set is empty. From the aforementioned argument, we get the minimum value of the objective function in this region by putting these values of $y$ in equation \eqref{eq_dmt_222_func_y} as
\begin{equation}
\label{solution_region_11}
d(r,\mathcal{R}_{1,1})=\min\left\{\begin{array}{c}
d_{2,2}(r)+d_{2,k}(r), ~0\leq r\leq 1;\\
k+3+(k+1)\left(\frac{2-3r}{2-r}\right),~0\leq r\leq \frac{2}{3};\\
k+6\left(\frac{1-r}{r}\right),~\frac{2}{3}\leq r\leq 1;
\end{array}\right.
\end{equation}

{\it Region $\mathcal{R}_{1,2}$:} In this region, the set of feasible values of $y$ is given by
\begin{equation*}
\left\{\frac{r}{2-r}\leq y\leq \frac{2r}{2-r}\right\}\cap \Big\{1\leq y\leq 2\Big\}\cap \Big\{r\leq y\Big\}.
\end{equation*}
For $\frac{2}{3}\leq r\leq 1$, $y_m=1$ and $y_M=\frac{2r}{2-r}$, for $1\leq r\leq \frac{4}{3}$, $y_m=\frac{r}{2-r}$ and $y_M=2$ and for $r\geq \frac{4}{3}$ the above set is empty. Putting these values of $y$ in equation \eqref{eq_dmt_222_func_y}, we get
\begin{equation}
\label{solution_region_12}
d(r,\mathcal{R}_{1,2})=\min\left\{\begin{array}{c}
k+6\left(\frac{1-r}{r}\right),~\frac{2}{3}\leq r\leq 1;\\
4+4(k-1)\left(\frac{1-r}{2-r}\right),~\frac{2}{3}\leq r\leq 1;\\
1+(k-1)\left(\frac{4-3r}{2-r}\right),~1\leq r \leq \frac{4}{3};\\
4\left(\frac{3-2r}{r}\right),~1\leq r \leq \frac{4}{3};
\end{array}\right.
\end{equation}

{\it Region $\mathcal{R}_{2,1}$:} In this region the set of feasible values of $y$ is given by
\begin{equation*}
\left\{0\leq y\leq \frac{r}{2-r}\right\}\cap \Big\{0\leq y\leq 1\Big\}\cap \Big\{r\leq y\Big\},
\end{equation*}
which is an empty set.

{\it Region $\mathcal{R}_{2,2}$:} In this region the set of feasible values of $y$ is given by
\begin{equation*}
\left\{0\leq y\leq \frac{r}{2-r}\right\}\cap \Big\{1\leq y\leq 2\Big\}\cap\Big\{r\leq y\Big\}.
\end{equation*}
This set is empty for $r\leq 1$; for $1\leq r\leq \frac{4}{3}$, $y_m=r$ and $y_M=\frac{r}{2-r}$, for $\frac{4}{3}\leq r\leq 2$, $y_m=r$ and $y_M=2$. Again, putting these values of $y$ in equation \eqref{eq_dmt_222_func_y}, we get
\begin{equation}
\label{solution_region_22}
d(r,\mathcal{R}_{2,2})=\min\left\{\begin{array}{c}
d_{2,2}(r)+d_{2,k}(r), ~1\leq r\leq 2;\\
1+(k-1)\left(\frac{4-3r}{2-r}\right),~1\leq r \leq \frac{4}{3};\\
2\left(\frac{2-r}{r}\right),~\frac{4}{3}\leq r \leq 2;
\end{array}\right.
\end{equation}
Finally, combining equations \eqref{solution_region_11}, \eqref{solution_region_12} and \eqref{solution_region_22} the theorem is proved.

\section{Proof of Lemma~\ref{lemma_eig1_help1}}
The determinant of a matrix can be written as the sum of several terms, where each term is a product of some of the elements of the matrix and $\pm 1$ (e.g., see Section $0.3.2$ in \cite{HJ}). In what follows, we shall first use this result implicitly to simplify $\triangle_1(\bar{\sigma},\bar{\xi})$. The terms of $\triangle_1(\bar{\sigma},\bar{\xi})$ will be gradually approximated in such a way that when expanded as a sum, the exponential order of each term in the sum remains unchanged, which in turn imply that the exponential order of $\triangle_1(\bar{\sigma},\bar{\xi})$ itself remains unchanged. Then combining it with the asymptotic expressions for the Vandermonde matrices $\mathbf{V}_1(\bar{\xi})$ and $\mathbf{V}_1(\bar{\sigma})$ we get the desired identity.

Replacing $\bar{x}$ by $\bar{\xi}$ and $\bar{y}$ by $\bar{\sigma}$ in the expression for $\triangle_1$ in equation (\ref{eq_del1}), we get
\begin{IEEEeqnarray*}{rl}
\triangle_1\left(\bar{\sigma},\bar{\xi}\right)=&\left|\begin{array}{cccccc}
1 & \sigma_{N_2} & \cdots \sigma_{N_2}^{(N_2-N_1-1)} & \sigma_{N_2}^{(N_2-N_1-1)}e^{-\frac{\xi_{N_1}}{\sigma_{N_2}}} & \cdots & \sigma_{N_2}^{(N_2-N_1-1)}e^{-\frac{\xi_1}{ \sigma_{N_2}}}\\
1 & \sigma_{(N_2-1)} & \cdots \sigma_{(N_2-1)}^{(N_2-N_1-1)} & \sigma_{(N_2-1)}^{(N_2-N_1-1)}e^{-\frac{\xi_{N_1}}{ \sigma_{(N_2-1)}}} & \cdots & \sigma_{(N_2-1)}^{(N_2-N_1-1)}e^{-\frac{\xi_1}{ \sigma_{(N_2-1)}}}\\
&&& \vdots &&\\
1 & \sigma_1 & \cdots \sigma_1^{(N_2-N_1-1)} & \sigma_1^{(N_2-N_1-1)}e^{-\frac{\xi_{N_1}}{ \sigma_1}} & \cdots & \sigma_1^{(N_2-N_1-1)}e^{-\frac{\xi_1}{ \sigma_1}}\\
\end{array}\right|,\\
=&T_0\left|\begin{array}{cccccc}
\sigma_{N_2}^{-(N_2-N_1-1)} \cdots  \sigma_{N_2}^{-1} & 1 & e^{-\frac{\xi_{N_1}}{\sigma_{N_2}}} & \cdots & e^{-\frac{\xi_1}{ \sigma_{N_2}}}\\
\sigma_{(N_2-1)}^{-(N_2-N_1-1)} \cdots \sigma_{(N_2-1)}^{-1} & 1 & e^{-\frac{\xi_{N_1}}{ \sigma_{(N_2-1)}}} & \cdots & e^{-\frac{\xi_1}{ \sigma_{(N_2-1)}}}\\
&&& \vdots &&\\
\sigma_1^{-(N_2-N_1-1)} \cdots \sigma_1^{-1} & 1 & e^{-\frac{\xi_{N_1}}{ \sigma_1}} & \cdots & e^{-\frac{\xi_1}{ \sigma_1}}\\
\end{array}\right|,
\end{IEEEeqnarray*}
where $T_0=\prod_{l=1}^{N_2}\left(\sigma_l^{(N_2-N_1-1)}\right)$. Now, putting $\sigma_i=(1+\rho\lambda_{(N_2-i+1)})^{-1}, ~\forall i$ in the above equation we get
\begin{IEEEeqnarray*}{rl}
\triangle_1\left(\bar{\sigma},\bar{\xi}\right)=&T_0\left|\begin{array}{cccccc}
(1+\rho\lambda_1)^{(N_2-N_1-1)} \cdots  (1+\rho\lambda_1) & 1 & e^{-\xi_{N_1}(1+\rho\lambda_1)} & \cdots & e^{-\xi_1(1+\rho\lambda_1)}\\
(1+\rho\lambda_2)^{(N_2-N_1-1)} \cdots (1+\rho\lambda_2) & 1 & e^{-\xi_{N_1}(1+\rho\lambda_2)} & \cdots & e^{-\xi_1(1+\rho\lambda_2)}\\
&&& \vdots &&\\
(1+\rho\lambda_{N_2})^{(N_2-N_1-1)} \cdots (1+\rho\lambda_{N_2}) & 1 & e^{-\xi_{N_1}(1+\rho\lambda_{N_2})} & \cdots & e^{-\xi_1(1+\rho\lambda_{N_2})}\\
\end{array}\right|,\\
=& ~T_1 \left|\begin{array}{cccccc}
1 & (1+\rho\lambda_1) & \cdots (1+\rho\lambda_1)^{(N_2-N_1-1)} & e^{-\xi_{N_1}(\rho\lambda_1)} & \cdots & e^{-\xi_1(\rho\lambda_1)}\\
1 & (1+\rho\lambda_2) & \cdots (1+\rho\lambda_2)^{(N_2-N_1-1)} & e^{-\xi_{N_1}(\rho\lambda_2)} & \cdots & e^{-\xi_1(\rho\lambda_2)}\\
&&& \vdots &&\\
1 & (1+\rho\lambda_{N_2})
 & \cdots (1+\rho\lambda_{N_2})^{(N_2-N_1-1)} & e^{-\xi_{N_1}(\rho\lambda_{N_2})} & \cdots & e^{-\xi_1(\rho\lambda_{N_2})}\\
\end{array}\right|,
\end{IEEEeqnarray*}
where $T_1=T_0\prod_{j=1}^{N_1}e^{-\xi_j}$. We have also ignored the sign change due to the row operation since $\triangle_1$ is a part of a pdf. To simplify the determinant in the above equation we do the following column operations: $C_i\to C_i-C_{i-1}(1+\rho \lambda_1), 2\leq i\leq (N_2-N_1)$ and $C_i\to C_i-C_1 e^{-\xi_{(i-N_2+N_1)} \rho \lambda_1}, (N_2-N_1+1)\leq i\leq N_2$. Since all the eigenvalues vary exponentially with $\rho$, for asymptotic $\rho$ using the approximation $\rho \xi_j(\lambda_i-\lambda_1)\dot{=}-\rho \xi_j \lambda_1$ and $\lambda_i-\lambda_1 \dot{=} -\lambda_1, \forall i\geq 2$ and $\forall ~ j$, we get
\begin{IEEEeqnarray*}{rl}
\triangle_1=&T_1 \left|\begin{array}{ccccc}
1 & 0 & \cdots 0 & \cdots & 0\\
1 & -\rho\lambda_1 & \cdots -(1+\rho\lambda_2)^{(N_2-N_1-2)}\rho\lambda_1 & \cdots & e^{-\xi_1(\rho\lambda_2)}\left(1-e^{-\xi_1\rho\lambda_1}\right)\\
&& \vdots &&\\
1 & -\rho\lambda_1
& \cdots -(1+\rho\lambda_{N_2})^{(N_2-N_1-2)}\rho\lambda_1 & \cdots & e^{-\xi_1(\rho\lambda_{N_2})}\left(1-e^{-\xi_1\rho\lambda_1}\right)\\
\end{array}\right|\\
=&T_2 \left|\begin{array}{cccccc}
1 & (1+\rho\lambda_2) & \cdots (1+\rho\lambda_2)^{(N_2-N_1-2)} & e^{-\xi_{N_1}(\rho\lambda_2)} & \cdots & e^{-\xi_1(\rho\lambda_2)}\\
&&& \vdots &&\\
1 & (1+\rho\lambda_{N_2})
& \cdots (1+\rho\lambda_{N_2})^{(N_2-N_1-2)} & e^{-\xi_{N_1}\rho\lambda_{N_2}}& \cdots & e^{-\xi_1(\rho\lambda_{N_2})}\\
\end{array}\right|,
\end{IEEEeqnarray*}
where $T_2=T_1(\rho\lambda_1)^{(N_2-N_1-1)}\prod_{j=1}^{N_1}\left(1-e^{-\rho \xi_j \lambda_1}\right)$. Proceeding in the same way we get
\begin{IEEEeqnarray}{rl}
\triangle_1=&\prod_{i=1}^{N_2-N_1}\left((\rho\lambda_i)^{(N_2-N_1-i)}\prod_{j=1}^{N_1}(1-e^{-\rho \xi_j\lambda_i})\right)T_1
\left|\begin{array}{ccc}
e^{-\rho \xi_{N_1}\lambda_{(N_2-N_1+1)}}  & \cdots & e^{-\xi_1\rho\lambda_{(N_2-N_1+1)}}\\
& \vdots &\\
e^{-\rho \xi_{N_1}\lambda_{N_2}} & \cdots & e^{-\xi_1\rho\lambda_{N_2}}\\
\end{array}\right| \nonumber \\
\label{pf_lemma_eig1_help2_final1}
= &\prod_{i=1}^{N_2-N_1}\left((\rho\lambda_i)^{(N_2-N_1-i)}\prod_{j=1}^{N_1}(1-e^{-\rho \xi_j\lambda_i})\right) T_2
\left|\begin{array}{ccc}
e^{-\xi_{N_1}(1+\rho \mu_1)}  & \cdots & e^{-\xi_1(1+\rho\mu_1)}\\
& \vdots &\\
e^{-\rho \xi{N_1}(1+\mu_{N_1})} & \cdots & e^{-\xi_1(1+\rho\mu_{N_1})}\\
\end{array}\right|,
\end{IEEEeqnarray}
where $\mu_l=\lambda_{N_2-N_1+l}, 1\leq l\leq N_1$ and $ T_2=\prod_{i=1}^{N_2}(1+\rho\lambda_i)^{-(N_2-N_1-1)}$. Now, denoting the determinant in equation \eqref{pf_lemma_eig1_help2_final1} by $\triangle_0$, we get


\begin{IEEEeqnarray}{rl}
\triangle_0\stackrel{(a)}{=}&\left|\begin{array}{cccc}
e^{-\xi_1(1+\rho \mu_{N_1})} & e^{-\xi_2(1+\rho\mu_{N_1})} & \cdots & e^{-\xi_{N_1}(1+\rho\mu_{N_1})}\\
&& \vdots &\\
e^{-\rho \xi_1(1+\mu_1)} & e^{-\xi_2(1+\rho\mu_1)}& \cdots & e^{-\xi_{N_1}(1+\rho\mu_1)}\\
\end{array}\right|,\\\dot{=}& \prod_{j=1}^{N_1}\left(e^{-\xi_j}e^{-(\rho \xi_j \mu_{(N_1+1-j)})}\right) \prod_{j=1}^{N_1}\prod_{i=1}^{(N_1-j)}\left(1-e^{-(\rho \xi_j \mu_i)}\right),
\end{IEEEeqnarray}
where equality $(a)$ is obtained by rearranging both the rows and columns in the reverse order and the last equality follows from Lemma~\ref{lemma_eig1_help1}. Using this expression and replacing the values of $\mu_i$'s in equation (\ref{pf_lemma_eig1_help2_final1}) we get
\begin{equation}
\label{eq_lemma3_temp1}
\triangle_1\dot{=}\prod_{i=1}^{N_2}(1+\rho\lambda_i)^{-(N_2-N_1-1)} \prod_{i=1}^{N_2-N_1}\left((\rho\lambda_i)^{(N_2-N_1-i)}\right) \prod_{j=1}^{N_1}\left(e^{-\xi_j}e^{-(\rho \xi_j \lambda_{(N_2+1-j)})}\right) \prod_{j=1}^{N_1}\prod_{i=1}^{(N_2-j)}\left(1-e^{-(\rho \xi_j \lambda_i)}\right)
\end{equation}
On the other hand, using equation \eqref{eq_asymptotic_def} at high SNR, we have
\begin{IEEEeqnarray}{l}
\label{eq_lemma3_temp2}
\mathbf{V}_1(\bar{\xi})\dot{=}\prod_{j=1}^{N_1}\xi_j^{(N_1-j)}=\left(\prod_{l=1}^{N_1}\xi_l^{(N_2-N_1)}\right)^{-1} \prod_{j=1}^{N_1}\xi_j^{(N_2-j)}.
\end{IEEEeqnarray}
\begin{IEEEeqnarray}{rl}
\mathbf{V}_1(\bar{\sigma})= & \det \left([(1+\rho\lambda_{(N_2-i+1)})^{-(j-1)}]_{i,j=1}^{N_2,N_2}\right),\\
= & \left(\prod_{i=1}^{N_2}(1+\rho\lambda_i)^{(1-N_2)}\right)\det \left([(1+\rho\lambda_{(N_2-i+1)})^{(N_2-j)}]_{i,j=1}^{N_2,N_2}\right),\nonumber \\
=&\left(\prod_{i=1}^{N_2}(1+\rho\lambda_i)^{(1-N_2)}\right)\det \left([(1+\rho\lambda_i)^{(j-1)}]_{i,j=1}^{N_2,N_2}\right),\nonumber \\
\label{eq_lemma3_temp3}
\dot{=} & \left(\prod_{i=1}^{N_2}(1+\rho\lambda_i)^{(1-N_2)}\right) \left(\prod_{i=1}^{N_2}(\rho\lambda_{(N_2-i+1)})^{(N_2-i)}\right).
\end{IEEEeqnarray}
Finally, combining equations \eqref{eq_lemma3_temp1}-\eqref{eq_lemma3_temp3} we have
\begin{IEEEeqnarray*}{rl}
\frac{\triangle_1(\bar{\sigma},\bar{\xi})}{\mathbf{V}_1(\bar{\xi}) \mathbf{V}_1(\bar{\sigma})}\dot{=} & \prod_{i=1}^{N_2}(1+\rho\lambda_i)^{N_1} \prod_{j=1}^{N_1}\left(\xi_j^{(N_2-N_1)}e^{-\xi_j}e^{-(\rho \xi_j \lambda_{(N_2+1-j)})}\right) \frac{\prod_{j=1}^{N_1}\prod_{i=1}^{(N_2-j)}\left(1-e^{-(\rho \xi_j \lambda_i)}\right)}{\left(\prod_{j=1}^{N_1}\xi_j^{(N_2-j)}\right)\left(\prod_{i=1}^{N_2}(\rho \lambda_i)^{(N_2-i)\land N_1}\right)},\\
= & \prod_{i=1}^{N_2}(1+\rho\lambda_i)^{N_1} \prod_{j=1}^{N_1}\left(\xi_j^{(N_2-N_1)}e^{-\xi_j}e^{-(\rho \xi_j \lambda_{(N_2+1-j)})}\right) \prod_{j=1}^{N_1}\prod_{i=1}^{(N_2-j)}\left(\frac{1-e^{-(\rho \xi_j \lambda_i)}}{\rho \xi_j \lambda_i}\right).
\end{IEEEeqnarray*}

\section{Proof of Lemma~\ref{lemma_eig1_help2}}
We shall simplify the term $D(\bar{\xi},\bar{\lambda})$ using a similar method as in the proof of Lemma~\ref{lemma_eig1_help1}. Let us start by expressing $D(\bar{\xi},\bar{\lambda})$ as follows
\begin{equation}
\label{eq_pf_help2_t1}
D(\bar{\xi},\bar{\lambda})=\prod_{i=1}^{N_2}\lambda_i \left|\begin{array}{ccc}
e^{-\xi_1\rho\lambda_{N_2}} & \cdots & e^{-\xi_{N_2}\rho\lambda_{N_2}}\\
& \vdots &\\
e^{-\xi_1\rho\lambda_1}& \cdots & e^{-\xi_{N_2}\rho\lambda_1}
\end{array}\right|\triangleq \left(\prod_{i=1}^{N_2}\lambda_i\right)\triangle.
\end{equation}
It is well known~\cite{HJ} that, for any square invertible matrix $A$ and square matrix $K$ the following identity holds.
\begin{equation*}
\left|\left[\begin{array}{cc}
A & B\\
C & K
\end{array}\right]\right|=|A||K-CA^{-1}B|.
\end{equation*}
To simplify $\triangle$ we substitute $A=e^{-\xi_1\rho \lambda_{N_2}}$ and use the above equation to get
\begin{equation*}
\triangle =e^{-\xi_1\rho\lambda_{N_2}}\left|\left[\begin{array}{ccc}
e^{-\xi_2\rho\lambda_{N_2-1}} & \cdots & e^{-\xi_{N_2}\rho\lambda_{N_2-1}}\\
& \vdots &\\
e^{-\xi_2\rho\lambda_1}& \cdots & e^{-\xi_{N_2}\rho\lambda_1}
\end{array}\right] - e^{\xi_1\rho\lambda_{N_2}} \left[\begin{array}{c}
e^{-\xi_1\rho\lambda_{N_2-1}} \\
 \vdots \\
e^{-\xi_1\rho\lambda_1}
\end{array}\right] \left[e^{-\xi_2\rho\lambda_{N_2}}  \cdots  e^{-\xi_{N_2}\rho\lambda_{N_2}}\right]\right|.
\end{equation*}
Since the eigenvalues $\xi_j$'s and $\lambda_i$'s vary exponentially with SNR ($\rho$), from the ordering among themselves and equation \eqref{eq_asymptotic_def} we have $\lambda_{N_2}-\lambda_i \dot{=}-\lambda_i, ~\forall i\leq (N_2-1)$. Using this in the above equation we get
\begin{IEEEeqnarray*}{rl}
\triangle ~\dot{=}&~e^{-\xi_1\rho\lambda_{N_2}}\left|\left[\begin{array}{ccc}
e^{-\xi_2\rho\lambda_{N_2-1}} & \cdots & e^{-\xi_{N_2}\rho\lambda_{N_2-1}}\\
& \vdots &\\
e^{-\xi_2\rho\lambda_1}& \cdots & e^{-\xi_{N_2}\rho\lambda_1}
\end{array}\right] -\left[\begin{array}{c}
e^{-\xi_1\rho\lambda_{N_2-1}} \\
 \vdots \\
e^{-\xi_1\rho\lambda_1}
\end{array}\right] \left[e^{-\xi_2\rho\lambda_{N_2}}  \cdots  e^{-\xi_{N_2}\rho\lambda_{N_2}}\right]\right|,\\
\stackrel{(a)}{\dot{=}}& e^{-\xi_1\rho\lambda_{N_2}}\left|\left[\begin{array}{ccc}
e^{-\xi_2\rho\lambda_{N_2-1}} & \cdots & e^{-\xi_{N_2}\rho\lambda_{N_2-1}}\\
& \vdots &\\
e^{-\xi_2\rho\lambda_1}& \cdots & e^{-\xi_{N_2}\rho\lambda_1}
\end{array}\right] -\left[\begin{array}{ccc}
e^{-\xi_1\rho\lambda_{N_2-1}} & \cdots & e^{-\xi_1\rho\lambda_{N_2-1}}\\
& \vdots &\\
e^{-\xi_1\rho\lambda_1} & \cdots & e^{-\xi_1\rho\lambda_1}
\end{array}\right] \right|, \\
\stackrel{(b)}{\dot{=}}& e^{-\xi_1\rho\lambda_{N_2}}\left|\left[\begin{array}{ccc}
e^{-\xi_2\rho\lambda_{N_2-1}}\left(1-e^{-\xi_1\rho \lambda_{N_2-1}}\right) & \cdots & e^{-\xi_{N_2}\rho\lambda_{N_2-1}}\left(1-e^{-\xi_1\rho \lambda_{N_2-1}}\right)\\
& \vdots &\\
e^{-\xi_2\rho\lambda_1}\left(1-e^{-\xi_1\rho \lambda_1}\right)& \cdots & e^{-\xi_{N_2}\rho\lambda_1}\left(1-e^{-\xi_1\rho \lambda_1}\right)
\end{array}\right]\right|,\\
\dot{=}& e^{-\xi_1\rho\lambda_{N_2}}\prod_{i=1}^{N_2-1}\left(1-e^{-\xi_1\rho \lambda_i}\right)\left|\left[\begin{array}{ccc}
e^{-\xi_2\rho\lambda_{N_2-1}} & \cdots & e^{-\xi_{N_2}\rho\lambda_{N_2-1}}\\
& \vdots &\\
e^{-\xi_2\rho\lambda_1}& \cdots & e^{-\xi_{N_2}\rho\lambda_1}
\end{array}\right]\right|,
\end{IEEEeqnarray*}
where step $(a)$ follows from the fact that $(\xi_1\lambda_i+\xi_2\lambda_{N_2})\dot{=}\xi_1\lambda_i,~
\forall i\leq (N_2-1)$ and step $(b)$ follows from the fact that $(\xi_1\lambda_i-\xi_j\lambda_i)\dot{=}\xi_1\lambda_i~
\forall j\geq 2, i$. Proceeding in the same way we get
\begin{equation*}
\triangle~\dot{=}~\prod_{j=1}^{N_2}e^{-\xi_j\rho\lambda_{N_2-j+1}}\prod_{i=1}^{N_2-j}\left(1-e^{-\xi_j\rho \lambda_i}\right).
\end{equation*}
Finally, using this asymptotic expression of $\triangle$ in equation \eqref{eq_pf_help2_t1}, we get
\begin{IEEEeqnarray*}{l}
D(\bar{\xi},\bar{\lambda})=\prod_{j=1}^{N_2}e^{-\xi_j} \triangle ~\dot{=}~\prod_{j=1}^{N_2}\left(e^{-\xi_j}e^{-\xi_j\rho\lambda_{N_2-j+1}}\prod_{i=1}^{N_2-j}\left(1-e^{-\xi_j\rho \lambda_i}\right)\right).\nonumber
\end{IEEEeqnarray*}

\bibliographystyle{IEEETran}
\bibliography{mybibliography,mybibliography-mv}

\begin{thebibliography}{10}
\providecommand{\url}[1]{#1}
\csname url@samestyle\endcsname
\providecommand{\newblock}{\relax}
\providecommand{\bibinfo}[2]{#2}
\providecommand{\BIBentrySTDinterwordspacing}{\spaceskip=0pt\relax}
\providecommand{\BIBentryALTinterwordstretchfactor}{4}
\providecommand{\BIBentryALTinterwordspacing}{\spaceskip=\fontdimen2\font plus
\BIBentryALTinterwordstretchfactor\fontdimen3\font minus
  \fontdimen4\font\relax}
\providecommand{\BIBforeignlanguage}[2]{{%
\expandafter\ifx\csname l@#1\endcsname\relax
\typeout{** WARNING: IEEEtran.bst: No hyphenation pattern has been}%
\typeout{** loaded for the language `#1'. Using the pattern for}%
\typeout{** the default language instead.}%
\else
\language=\csname l@#1\endcsname
\fi
#2}}
\providecommand{\BIBdecl}{\relax}
\BIBdecl

\bibitem{relayTG}
``I. 802.16s relay task group,'' \texttt{http://www.ieee802.org/16/relay/}.

\bibitem{SEB1}
A.~Sendonaris, E.~Erkip, and B.~Aazhang, ``User cooperation diversity-part i:
  System description,'' \emph{IEEE Transactions on Communications}, vol.~51,
  pp. 1927--1938, Nov, 2003.

\bibitem{SEB2}
------, ``User cooperation diversity-part ii: Implementation aspects and
  performance analysis,'' \emph{IEEE Transactions on Communications}, vol.~51,
  pp. 1939--1948, Nov, 2003.

\bibitem{LW}
J.~N. Laneman and G.~Wornell, ``Distributed spce-time-coded protocols for
  exploiting cooperative diversity in wireless networks,'' \emph{IEEE
  Transactions on Information Theory}, vol.~49, pp. 2415--2425, Oct, 2003.

\bibitem{LanemanJN:coop}
J.~N. Laneman, D.~N.~C. Tse, and G.~W. Wornell, ``Cooperative diversity in
  wireless networks: Efficient protocols and outage behavior,'' \emph{IEEE
  Transactions on Information Theory}, vol.~50, no.~12, pp. 3062--3080, Dec.
  2003.

\bibitem{NabarRU:relay}
R.~U. Nabar, H.~Bolcskei, and F.~W. Kneubuhler, ``Fading relay channels:
  Performance limits and space-time signal design,'' \emph{Journ. Selec. Areas
  Commun.}, vol.~22, no.~6, pp. 1099--1109, Aug. 2004.

\bibitem{PrasadN:CO-OP:ISIT04}
N.~Prasad and M.~K. Varanasi, ``Diversity and multiplexing tradeoff bounds for
  cooperative diversity schemes,'' in \emph{Proc. IEEE Intl. Symp. Inform.
  Th.}, Chicago, IL, Jun. 2004.

\bibitem{KHP}
K.~Azarian, H.~E. Gamal, and P.~Schniter, ``On the achievable
  diversity-multiplexing tradeoff in half-duplex cooperative channels,''
  \emph{IEEE Transactions on Information Theory}, vol.~51, pp. 4152--4172, Dec,
  2005.

\bibitem{NpV}
N.~Prasad and M.~K. Varanasi, ``High performance static and dynamic cooperative
  communication protocols for the half duplex fading relay channel,'' in
  \emph{Proceedings of Global Telecommunications Conference, San Fransisco},
  Nov-Dec, 2006, pp. 1--5.

\bibitem{tse1}
L.~Zheng and D.~Tse, ``Diversity and multiplexing: A fundamental tradeoff in
  multiple antenna channels,'' \emph{IEEE Transactions on Information Theory},
  vol.~49, pp. 1073--1096, May, 2003.

\bibitem{YEE}
M.~Yuksel and E.~Erkip, ``Multi-antenna cooperative wireless systems: A
  diversity multiplexing tradeoff perspective,'' \emph{IEEE Transactions on
  Information Theory, Special Issue on Models, Theory and Codes for Relaying
  and Cooperation in Communication Networks}, vol.~53, pp. 3371--3393, Oct,
  2007.

\bibitem{RC}
K.~R. Kumar and G.~Caire, ``Coding and decoding for the dynamic decode and
  forward relay protocol,'' \emph{IEEE Transactions on Information Theory},
  vol.~55, pp. 3186--3205, Jul, 2009.

\bibitem{lvy_2009}
O.~Leveque, C.~Vignat, and M.~Yuksel, ``Diversity-multiplexing tradeoff for the
  mimo static half-duplex relay,'' in \emph{Proc. IEEE Int. Symphosium on
  Information Theory, Seoul, South Korea}, Jun, 2009.

\bibitem{skv_2009}
S.~Karmakar and M.~K. Varanasi, ``Diversity-multiplexing tradeoff of the
  dynamic decode and forward protocol on a mimo half-duplex relay channel,'' in
  \emph{Proc. IEEE Int. Symphosium on Information Theory, Seoul, South Korea},
  Jun, 2009.

\bibitem{SkV2}
------, ``Optimal dmt of dynamic decode-and-forward protocol on a half-duplex
  relay channel with arbitrary number of antennas at each node,'' in
  \emph{Proceedings of Asilomar Conf. on Signals, Systems and Computers,
  Pacific Grove, CA}, Nov, 2009.

\bibitem{OCM}
O.~Leveque, C.~Vignat, and M.~Yuksel, ``Diversity-multiplexing tradeoff for the
  mimo static half-duplex relay,'' Dec, 2008, preprint, available at
  \texttt{http://www.arxiv.org}.

\bibitem{KHP1}
K.~Azarian, H.~E. Gamal, and P.~Schniter, ``On the optimality of the arq-ddf
  protocol,'' \emph{IEEE Transactions on Information Theory}, vol.~54, pp.
  1718--1724, Apr, 2008.

\bibitem{GaS}
H.~Gao and P.~J. Smith, ``A determinant representation for the distribution of
  quadratic forms in complex normal vectors,'' \emph{Journ. of Multivariate
  Analysis}, vol.~73, pp. 155--165, Feb, 2000.

\bibitem{CZZ}
M.~Chiani, M.~Z. Win, and A.~Zanella, ``On the capacity of spatially correlated
  mimo rayleigh-fading channels,'' \emph{IEEE Transactions on Information
  Theory}, vol.~49, pp. 2363--2371, Oct, 2003.

\bibitem{Sanjay_Varanasi_ZIC_DMT}
S.~Karmakar and M.~K. Varanasi, ``The diversity-multiplexing tradeoff of the
  mimo z interference channel,'' in \emph{Proc. IEEE Int. Symphosium on
  Information Theory, Austin, Texas}, Jun, 2010.

\bibitem{CT}
T.~M. Cover and J.~A. Thomas, \emph{Elements of Information Theory}.\hskip 1em
  plus 0.5em minus 0.4em\relax Wiley, 1991.

\bibitem{TB}
G.~Taricco and E.~Biglieri, ``Exact pairwise error probability of space-time
  codes,'' \emph{IEEE Transactions on Information Theory}, vol.~48, pp.
  510--513, Feb, 2002.

\bibitem{Sanjay_Varanasi_Symmetric_HDRC_DMT}
S.~Karmakar and M.~K. Varanasi, ``The diversity-multiplexing tradeoff of the
  symmetric mimo half-duplex relay channel,'' in \emph{Proc. IEEE Int.
  Symphosium on Information Theory, Austin, Texas}, Jun, 2010.

\bibitem{YHXM}
Y.~Yang, H.~Hu, J.~Xu, and G.~Mao, ``Relay technologies for wimax and
  lte-advanced mobile systems,'' \emph{IEEE Communications Magazine}, vol.~47,
  pp. 100--105, Oct, 2009.

\bibitem{WSC}
W.~Wei, V.~Srinivasan, and K.-C. Chua, ``Using mobile relays to prolong the
  lifetime of wireless sensor networks,'' in \emph{Proceedings of ACM MobiCom
  2005, Cologne, Germany}, August 2005.

\bibitem{PAT}
S.~Pawar, A.~Avestimehr, and D.~Tse, ``Diversity-multiplexing tradeoff of the
  half-duplex relay channel,'' in \emph{Proceedings of 46th Annual Allerton
  Conference on Communication, Control, and Computing}, Sept. 2008, pp. 27--33.

\bibitem{Tulino_Verdu}
A.~M. Tulino and S.~Verdu, \emph{Random Matrix Theory and Wireless
  Communications}.\hskip 1em plus 0.5em minus 0.4em\relax Now, 2004.

\bibitem{HJ}
R.~A. Horn and C.~R. Jhonson, \emph{Matrix analysis}.\hskip 1em plus 0.5em
  minus 0.4em\relax Cambridge Univ. Press, 1990, vol. 1st.

\bibitem{SMM}
S.~H. Simon, A.~L. Moustakas, and L.~Marinelli, ``Capacity and character
  expansions: Moment-generating function and other exact results for mimo
  correlated channels,'' \emph{IEEE Transactions on Information Theory},
  vol.~52, pp. 5336--5351, Dec, 2006.

\end{thebibliography}

\end{document}